\documentclass[aps,pra,amsmath,amssymb,showpacs,twocolumn]{revtex4-1}

\usepackage[T1]{fontenc}
\usepackage{amssymb,amsmath}
\usepackage{amsfonts}
\usepackage{graphicx,epsfig}
\usepackage{amsthm}
\usepackage{color}
\usepackage{mathbbol}
\usepackage{epstopdf}
\usepackage{booktabs}
\usepackage{multirow}

\def\ket#1{|#1\rangle}

\def\bra#1{\langle#1|}
\def\ketbra#1{|#1\rangle\langle#1|}

\def\tr{\mathrm{tr}}

\newcommand{\bn}{{\mathbf n}}

\newcommand{\bOne}{\mbox{\bf 1}}
\newcommand{\bsig}{\mbox{\boldmath{$\sigma$}}}

\newtheorem{theorem}{Theorem}[]
\newtheorem{definition}[]{Definition}

\newtheorem{lemma}[]{Lemma}

\usepackage{braket}

\begin{document}
\newcommand{\ri}{{\rm i}}
\newcommand{\re}{{\rm e}}
\newcommand{\bU}{{\bf U}}
\newcommand{\bd}{{\bf d}}
\newcommand{\be}{{\bf e}}
\newcommand{\br}{{\bf r}}
\newcommand{\bk}{{\bf k}}
\newcommand{\bE}{{\bf E}}
\newcommand{\bI}{{\bf I}}
\newcommand{\bR}{{\bf R}}
\newcommand{\bC}{{\bf C}}
\newcommand{\cL}{{\cal L}}
\def\Jp#1{J_+^{(#1)}}
\def\Jm#1{J_-^{(#1)}}
\newcommand{\bZero}{{\bf 0}}
\newcommand{\bM}{{\bf M}}
\renewcommand{\bm}{{\bf m}}
\newcommand{\bs}{{\bf s}}

\sloppy

\title{Quantum-phase synchronization}
\author{Lukas J. Fiderer$^{1}$, Marek Ku\'s$^2$ and Daniel Braun$^{1}$ }
\affiliation{$^1$ Eberhard-Karls-Universit\"at T\"ubingen, Institut
  f\"ur Theoretische Physik, 72076 T\"ubingen, Germany\\
$^{2}$ Center for Theoretical Physics of the Polish
  Academy of Sciences, 02-068 Warsaw, Poland}

\begin{abstract} We study mechanisms that allow one to
  synchronize the quantum phase of two qubits relative to a fixed
  basis. Starting from one qubit in a fixed reference state and
  the other in an unknown state, we find that contrary to the
  impossibility of perfect quantum
  cloning, the quantum-phase can be synchronized perfectly
  through a joined unitary
  operation.  When {\em both} qubits are initially in a pure unknown state,
  perfect quantum-phase synchronization through unitary operations
  becomes impossible.  In this situation we determine the maximum
  average quantum-phase synchronization fidelity, the distribution of
  relative phases and fidelities,
  and  identify optimal quantum
  circuits that achieve this maximum fidelity.  A subset
  of these optimal quantum circuits
  enable perfect quantum-phase synchronization for a class 
  of unknown
initial states restricted to
  the equatorial plane of the Bloch sphere.
\end{abstract}

\maketitle

{\em Introduction.} The quantum mechanical phase marks arguably the
most profound
deviation of quantum mechanics from classical mechanics.  It is at the
base of all quantum mechanical interference as displayed e.g. in the
double slit experiment, and, in the case of
multi-particle systems, enhanced correlations compared to the
classical world, as described by quantum mechanical entanglement.
Some of the most spectacular quantum mechanical
effects occur when phase coherence is
established over a macroscopic number of constitutents, as is the case e.g.~for
superconductivity \cite{buckel2012supraleitung}, superfluidity, Bose-Einstein condensates \cite{pitaevskii2016bose}, quantum
magnets \cite{auerbach2012interacting}, or lasing \cite{henry1982theory}. While the mechanisms that lead to synchronized
quantum phases are well
understood in these examples, one may ask what are the mechanisms in
general that allow one to synchronize the quantum phases of different
systems. Having an answer to that question might enable new types
of macroscopic quantum effects that we are not aware of yet.
In this paper we study quantum-phase synchronization (QPS) for the simplest
possible example, namely two qubits and unitary propagation.

We emphasize that the effects sought here are very different from
those in the field of quantum mechanics of systems that
classically synchronize, also called quantum (stochastic) synchronization
\cite{goychuk_quantum_2006,zhirov_synchronization_2008,shim_synchronized_2007}. In
those systems, one considers the periodic dynamics of
(typically driven) oscillators with slightly different frequencies
which under slight interaction give rise to a common dynamical mode in
which all oscillators synchronize, and research has mainly examined
the question to what extent quantum fluctuations affect that
synchronization when the oscillators become microscopic. QPS, on the
contrary, has no classical analog, as it
concerns the quantum phase which is only defined in the quantum
world. Recently, synchronization of an ensemble of interacting dipoles
modeled as qubits
was studied in \cite{zhu_synchronization_2015}. However, a fixed dipole
interaction was considered, whereas here we are interested in finding
the SU(4) joint-evolution that leads to
QPS.
QPS is related to quantum cloning
\cite{Gisin97,Bruss98,Werner98,fan_quantum_2014},
and it has been proposed \cite{DeMartini98} and experimentally
demonstrated \cite{Nagali07} that
quantum cloning can  amplify entanglement to a macroscopic level.  However, 
there are two crucial differences between QPS and quantum cloning:
{\em i.)} We want to synchronize only the 
quantum phase of the state, not the full state
itself. This implies a different target function (see below). {\em ii.)}
While cloning aims at attaining concurrence
of each output with the input state, QPS solely intends
to achieve concurrence among the outputs but not with an initial
state.\\

An unknown quantum state cannot be cloned perfectly
\cite{dieks_communication_1982,wootters_single_1982}. This remains
true even when restricting the set of input states, e.g.~to states in the
equatorial plane as in phase-covariant cloning (PCC)
\cite{brus_phase-covariant_2000,fan_quantum_2001,dariano_optimal_2003,fiurasek_optical_2003,de_chiara_quantum_2004,sciarrino_realization_2005,buscemi_economical_2005,cernoch_experimental_2006,Chen07,yao_multiple_2014} or its generalizations \cite{DAriano01,karimipour_generation_2002,du_experimental_2005}.
But here we show 
that quantum phases {\em can} be perfectly synchronized in the
standard situation of quantum cloning, where one has one qubit in a
known initial (blank) pure state, and the qubit to be copied in an unknown
state. We then ask for more and consider both initial
states as unknown. We show that perfect QPS is not possible anymore in this
 situation, but find quantum circuits that achieve maximal average
 QPS fidelity. 

{\em Phase synchronization fidelity.} A general pure state of a qubit
(spin-1/2) can be written in a fixed computational basis
$\{|0\rangle,|1\rangle\}$ as
\begin{equation}
  \label{eq:pure}
|\psi\rangle=\cos(\theta/2)|0\rangle+e^{i\varphi}\sin(\theta/2)|1\rangle,
\end{equation}
where $\varphi\in [0,2\pi]$ is the quantum phase of the
state, i.e.~the relative phase between the two basis states,  and
$\theta\in [0,\pi]$ defines the relative weight of the two basis
states. In the
corresponding density matrix,   $\rho=(\bOne+\bn\cdot\bsig)/2$,
$\varphi$ is coded in the azimuthal angle of the Bloch vector
$\bn=(\sin\theta\,\cos\varphi,\sin\theta\,\sin\varphi,\cos\theta)\equiv(n_x,n_y,n_z)$,
and $\bsig=(\sigma_x,\sigma_y,\sigma_z)$ is the vector of Pauli
matrices.
For mixed
states, we still consider $\varphi$ as given by the Bloch vector the
quantum phase of the state, as $\varphi$
still determines the oscillatory behavior
of expectation values,
e.g.~when $\varphi$ evolves linearly with time and one measures
$\langle\sigma_x\rangle$. Only the contrast of the oscillations is
reduced due to the admixture of the identity.
Thus, two states have the same quantum mechanical phase if the $xy$
components of
their Bloch vectors are aligned.  We therefore define quantum-phase fidelity
between two states with $||\bm_i||\ne 0$ as
\begin{equation}
  \label{eq:fidphase}
  f(\rho_1,\rho_2)\equiv
  \bm_1\cdot\bm_2/(||\bm_1||||\bm_2||)=\cos(\Delta \varphi)\,,
\end{equation}
where $||.||$ denotes the standard vector norm, $\bm_i$ the
projection of the Bloch vector of qubit $i$ into the $xy$ plane
($\bm=(n_x,n_y)$), and $\Delta\varphi=\varphi_2-\varphi_1$. If $||\bm_i||=0$,
the phase of qubit $i$, and therefore also the relative phase and
$f(\rho_1,\rho_2)$ are
undefined in this basis.

Consider a general linear quantum channel $\Phi$ on
the 2 qubits, i.e.~a completely positive map  $\mathcal M_4(\mathbb
C)\to\mathcal M_4(\mathbb C)$ that maps density matrices to density
matrices. Starting from an initial product state $\rho=\rho_1\otimes\rho_2$,
we obtain a final state $\rho'=\Phi\rho$ and reduced states
$\rho_1'=\tr_2\rho'$, $\rho_2'=\tr_1\rho'$. We
define the phase
synchronization fidelity  (PSF) as
\begin{align}
F(\rho,\Phi)&\equiv f(\rho_1',\rho_2')\,.
\end{align}
With this
definition, $F(\rho,\Phi)\in[-1,1]$, and $F=1$ ($F=-1$) corresponds to
perfect synchronization (perfect anti-synchronization) of the
quantum-phase for initial state $\rho$.

\begin{definition}\label{def:perfectQPS}
Quantum phase synchronization is said to be {\em perfect} for a
two-qubit quantum
channel $\Phi$ and a set of initial states ${\mathcal A}$, if $\forall
\rho\in
{\mathcal A}$ for which $F(\rho,\Phi)$ is defined,
$F(\rho,\Phi)=1$.
\end{definition}
When allowing arbitrary channels, perfect QPS,
i.e.~$F(\rho,\Phi)=1 \,\forall \, \rho$, can be achieved trivially by resetting
both qubits to the same state.  Therefore, optimizing QPS over
arbitrary quantum channels is not very interesting. Another example,
how non-unitary channels can achieve perfect QPS is the well-known
optimal cloning  machine 
of Bu\v{z}ek and Hillery \cite{buzek_quantum_1996}. It leads to perfect phase
synchronization as both final reduced states are identical. Similarly,
one easily sees that LOCC operations allow synchronized resetting of
the two states.  
To avoid such trivial constructions, we therefore restrict ourselves 
to unitary channels 
$\Phi_U:\rho\mapsto U\rho U^\dagger$, and write $F(\rho,U)\equiv
F(\rho,\Phi_{U})$. Note that unitary operations are also important from a practical perspective, when one tries to keep quantum coherence as long as possible (including, e.g. through error correction). \\

{\em One qubit in a known state.} First consider the standard initial state for
quantum state cloning: $\rho =|\psi_1\rangle\langle\psi_1|\otimes
|0\rangle\langle 0|$, i.e.~the first qubit is in an unknown pure state
$\rho_{1,p}\equiv|\psi_1\rangle\langle\psi_1|$, and the second in a
known (blank) state
$|0\rangle$. No linear transformation exists
that transforms $\rho$ such that at the output both qubits are in
state $|\psi_1\rangle$
\cite{dieks_communication_1982,wootters_single_1982}.  However, one
easily shows that
perfect QPS can be achieved, $F(\rho,U)=1$,
for all $\rho$ of the
above form. We use the following little lemma:
\begin{lemma}\label{lem:1}
Let $V_1,V_2$ be arbitrary single-qubit unitaries
acting on an arbitrary (possibly entangled) two qubit state,
$V_1=R_{\hat{\textbf{n}}}(\alpha)$,
$V_2=R_{\hat{\textbf{m}}}(\beta)$  and
$R_{\hat{\textbf{n}}}(\alpha)\equiv
e^{-i\frac{\alpha}{2}\hat{\textbf{n}}\cdot\boldsymbol{\sigma}}
$. Then for an
arbitrary density matrix $\rho\in\mathcal{D}_4$
(positive-semidefinite Hermitian matrices
with trace one) and a local
transformation $V_1\otimes V_2$,
 the Bloch vectors of the reduced density matrices corresponding to
 the propagated state
\begin{equation}
\rho^\prime=V_1\otimes V_2\rho (V_1\otimes V_2)^\dagger
\end{equation}
are given by the rotated initial Bloch vectors,
\begin{align}
\textbf{n}_1^\prime &= \tilde{R}_{\hat{\textbf{n}}}(\alpha)\,\textbf{n}_1\\
\textbf{n}_2^\prime &= \tilde{R}_{\hat{\textbf{m}}}(\beta)\,\textbf{n}_2\,,
\end{align}
\end{lemma}

\begin{proof}
We first show that the two partial traces that emerge from taking the
partial state and calculating the Bloch vector can be combined to a
trace over the whole system. For an arbitrary $\rho\in\mathcal{D}_4$
we find in the computational basis $\{\ket{0_1},\ket{1_1}\}$ for the
first qubit and $\{\ket{0_2},\ket{1_2}\}$ for the second qubit
\begin{align}
\tr(\boldsymbol{\sigma}\otimes\mathbb{1}\,\rho)&=\sum_{i_1,j_2=0}^1\bra{i_1}\bra{j_2}\boldsymbol{\sigma}\otimes\mathbb{1}\,\rho\ket{i_1}\ket{j_2}\\
&=\sum_{i_1=0}^1\bra{i_1}\sum_{j_2=0}^1\bra{j_2}\boldsymbol{\sigma}\otimes(\sum_{k_2=0}^1\ket{k_2}\bra{k_2})\rho\ket{j_2}\ket{i_1}\label{prop1}\\
&=\sum_{i_1=0}^1\bra{i_1}(\boldsymbol{\sigma}\,\sum_{k_2=0}^1\bra{k_2}\rho\ket{k_2}\ket{i_1}\\
&=\tr_1(\boldsymbol{\sigma}\,\tr_2(\rho)).
\end{align}
This is used in the following, where we prove the lemma for the first
Bloch vector,
\begin{align}
\textbf{n}_1^\prime &=\tr_1(\boldsymbol{\sigma}\,\tr_2(R_{\hat{\textbf{n}}}(\alpha)\otimes R_{\hat{\textbf{m}}}(\beta)\rho(R_{\hat{\textbf{n}}}(\alpha)\otimes R_{\hat{\textbf{m}}}(\beta))^\dagger))\\
&=\tr(\boldsymbol{\sigma}\otimes\mathbb{1}\,R_{\hat{\textbf{n}}}(\alpha)\otimes R_{\hat{\textbf{m}}}(\beta)\rho(R_{\hat{\textbf{n}}}(\alpha)\otimes R_{\hat{\textbf{m}}}(\beta))^\dagger)\label{prop2}\\
&=\tr((R_{\hat{\textbf{n}}}(\alpha)\otimes R_{\hat{\textbf{m}}}(\beta))^\dagger\boldsymbol{\sigma}\otimes\mathbb{1}\,R_{\hat{\textbf{n}}}(\alpha)\otimes R_{\hat{\textbf{m}}}(\beta)\rho)\label{prop4}\\
&=\tr((R_{\hat{\textbf{n}}}^\dagger(\alpha)\boldsymbol{\sigma}R_{\hat{\textbf{n}}}(\alpha))\otimes\mathbb{1}\rho)\\
&=\tr_1(R_{\hat{\textbf{n}}}^\dagger(\alpha)\boldsymbol{\sigma}R_{\hat{\textbf{n}}}(\alpha)\tr_2(\rho))\label{prop3}\\
&=\tr_1(\boldsymbol{\sigma}R_{\hat{\textbf{n}}}(\alpha)\tr_2(\rho)R_{\hat{\textbf{n}}}^\dagger(\alpha))\label{prop5}\\
&=\tilde{R}_{\hat{\textbf{n}}}(\alpha)\textbf{n}_1.\label{prop6}
\end{align}
In lines \ref{prop2} and \ref{prop3}
$\tr_1(\boldsymbol{\sigma}\,\tr_2(\rho))=\tr(\boldsymbol{\sigma}\otimes\mathbb{1}\rho)$
is used, that was shown above. In lines \ref{prop4} and \ref{prop5}
the cyclic property of the trace was used. Line \ref{prop6} uses the
rotation operator's property
\begin{equation}
\textbf{n}^\prime=\tr(\boldsymbol{\sigma}\,R_{\hat{\textbf{n}}}(\alpha)\rho R^\dagger_{\hat{\textbf{n}}}(\alpha))=\tilde{R}_{\hat{\textbf{n}}}(\alpha)\textbf{n},\label{rotoperator}\,.
\end{equation}
The proof for the second Bloch vector
can be done in analogous fashion.
\end{proof}
Actually, this Lemma also holds for $n$ qubits, i.e. for arbitrary $\rho\in\mathcal{D}_{2^n}$, where we find
\begin{equation}
\tr(\mathbb{1}_{2^{l-1}}\otimes\boldsymbol{\sigma}_l\otimes\mathbb{1}_{2^{n-l}}\,\rho)=\tr_l(\boldsymbol{\sigma}_l\tr_{1,...,l-1,l+1,...,n}(\rho)),
\end{equation}
where the trace operator's subscripts indicate the qubits that are
traced out. Then, the proof goes in a manner analogue to the lemma,
and one finds for the $l$th Bloch vector
\begin{equation}
\textbf{n}_l^\prime=\tr_l(\boldsymbol{\sigma}_lR_{\hat{\textbf{n}}}(\alpha)\tr_{1,...,l-1,l+1,...,n}(\rho)R_{\hat{\textbf{n}}}^\dagger(\alpha)),
\end{equation}
which, according to the definition of the rotation operator, implies that we may simply rotate the lth Bloch vector, $\textbf{n}_l^\prime=\tilde{R}_{\hat{\textbf{n}}}(\alpha)\textbf{n}_l$.

Now consider a propagation with a controlled-NOT (CNOT) operation with
qubit 1 as control and qubit 2 as target \cite{nielsen_2000_quantum}.
We denote by $C_{ij}$ a CNOT operation with $i$ as the controlling
and $j$ as the controlled qubit, hence $U=C_{12}$. Then one easily obtains
$\textbf{n}_1^\prime=\textbf{n}_2^\prime=\cos{\theta_1}\hat{\textbf{e}}_z.$
Thus, the initial phase of qubit 1 is erased, but both Bloch vectors are
aligned to the $z-$axis and identical to each other.  Now act with an arbitrary
local transformation that rotates the Bloch vectors away from the $z$-axis,
$
W_{\text{sync}}\equiv R_\textbf{n}(\alpha)\otimes R_\textbf{n}(\alpha).
$
Apart from the
case where after $C_{12}$ both qubits ended up in the maximally mixed
state,
i.e.~for $\theta_1=\frac{\pi}{2}$, the PSF equals
one, $F(\rho_{1,p}\otimes|0\rangle\!\langle
0|,W_{\text{sync}}C_{12})=1$. After the operation
$W_{\text{sync}}C_{12}$, the two qubits 
are perfectly phase synchronized. Thus, inspite of the fact that QPS
has something inherently 
irreversible (different initial phases are mapped to the same final phase),
tracing out a
qubit introduces enough irreversibility to the unitary evolution for
perfect QPS be possible if one qubit is initially known. There is no
contradiction to the result from PCC 
\cite{brus_phase-covariant_2000} that restriction to equatorial input
states for $1\to 2$ cloning yields a 
maximum fidelity of $1/2+\sqrt{1/8}\simeq 0.854$, as the fidelity
maximized there is the overlap between initial and final states, not PSF. 
Our result can be extended to an
initially mixed state of qubit 1. Furthermore, perfect QPS is easily
extended from one to $n-1$ qubits in blank states:

For initial states of the form $\rho=\rho_1\otimes\ket{0}_2\bra{0}_2\otimes\dotsm\otimes\ket{0}_n\bra{0}_n=\rho_1\otimes\ket{0\dotsm0}\bra{0\dotsm0}$ where $\rho_1=(1-p)\mathbb{1}/2+p\ket{\psi}\bra{\psi}$ with $0\leq p\leq1$ and an arbitrary pure state $\ket{\psi}$, see Eq. (1), the transformation $U=W_{\text{sync},n}C_{12}...C_{1n}$ achieves perfect QPS $f(\rho_1',\rho_i')=1\forall i\in\{2,\dotsc,n\}$ as follows from direct calculation: \\

Applying $C\equiv C_{12}...C_{1n}$ to the inital states one finds
\begin{align}
\rho^\prime &= C\rho_1\otimes\ket{0...0}\bra{0...0}C\\
&= \frac{1}{2}(1-p)C(\mathbb{1}_2\otimes\ket{0...0}\bra{0...0})C+\notag\\
&\qquad p\,C(\ket{\psi}\bra{\psi}\otimes\ket{0...0}\bra{0...0})C\\
&=\frac{1-p}{2}(\ket{0...0}\bra{0...0}+\ket{1...1}\bra{1...1})+p\,\ket{\psi_n}\bra{\psi_n},\label{transmix}
\end{align}
where $|\psi_n\rangle=\cos(\theta/2)|0\dotsm0\rangle+e^{i\varphi}\sin(\theta/2)|1\dotsm1\rangle$.
Because Pauli matrices are traceless, calculation of the Bloch vectors
gives zero for the first term in equation (\ref{transmix}). For the second term we obtain identical Bloch
vectors of all qubits in the final state,
      \begin{equation}
      \textbf{n}_i^\prime=p\left(\begin{array}{c}0\\0\\\cos{\theta_1}\\\end{array}\right), i=1,...,n.
      \end{equation}
$W_{\text{sync},n}\equiv R_\textbf{n}(\alpha)\otimes ...\otimes R_\textbf{n}(\alpha)$ rotates all $n$ Bloch vectors away from the z-axis. This holds as Lemma 1 generalizes to $n$ qubits. Apart from the cases of maximally mixed states, $\theta_1=\frac{\pi}{2}$ or $p=0$, phases are defined and perfectly synchronized, $F(\rho_1\otimes\ket{0...0}\bra{0...0},W_{\text{sync},n}C)=1$.\\

{\em Both qubits in unknown states.}  We now
attempt the more ambitious task of phase-synchronizing qubits
that are both initially in unknown pure states,
$\rho=\rho_{1,p}\otimes\rho_{2,p}$ with $\rho_{i,p}=\ketbra{\psi_i}$,
$i=1,2$, and $\psi_i$ of the form (\ref{eq:pure}).  For this we look
at the most general transformations $U$ of two qubits.
The set of all unitaries $U\in\mathcal{SU}(4)$ on two
qubits can be broken down into CNOT operations and single qubit
unitaries \cite{Deutsch95,DiVincenzo95,Barenco95b,Barenco95a,nielsen_2000_quantum}. The
parametrization of $\mathcal{SU}(4)$ requires $15$ real parameters.
Khaneja et al. \cite{khaneja_time_2001} as well as Kraus and Cirac
\cite{kraus_optimal_2001}, found a decomposition of an arbitrary $U$
$\mathcal{SU}(4)$ of the form
\begin{equation}
U=WU_cV\,,
\end{equation}
where $V=V_1\otimes V_2,W=W_1\otimes
W_2\in\mathcal{SU}(2)\,\times\,\mathcal{SU}(2)$ are \textit{local}
unitary transformations exclusively acting on each qubit
separately. $U_c$ is an element of the quotient space
$\mathcal{SU}(4)/\mathcal{SU}(2)\times\mathcal{SU}(2)$.

Minimal circuits for $U_c$ were reported in
\cite{bullock_arbitrary_2003,vidal_universal_2004,vatan_optimal_2004}.
We use
the circuit from Vatan and Williams (theorem 5 in \cite{vatan_optimal_2004})
according to which a general unitary
transformation can be written as
\begin{eqnarray}
U&=& (W_1\otimes W_2)U_c\,(V_1\otimes V_2)\\
U_c&\equiv&\,C_{21}\,(\mathbb{1}_2\otimes R_y(\beta)) \,C_{12}\,
(R_z(\gamma)\otimes R_y(\alpha)) \,C_{21} \label{eq:Uc}
\end{eqnarray}
 This leads to figure \ref{wholecircuit} for a general
 circuit for two qubits.
 \begin{figure}
 \includegraphics[width=0.7\linewidth]{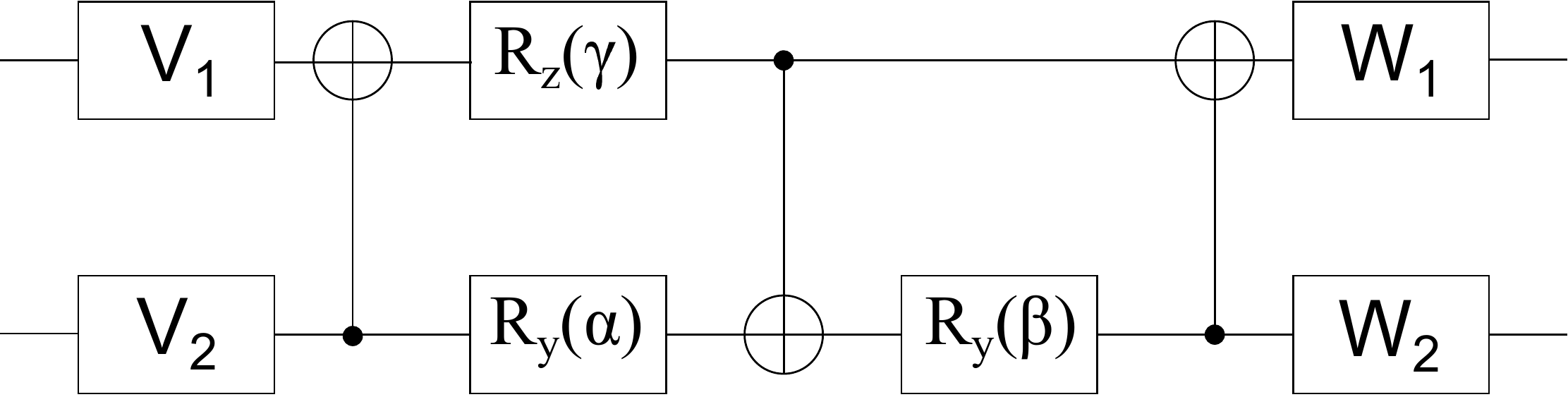}
   \caption{(color online). Quantum circuit for the most general unitary transformation
     $U\in\mathcal{SU}(4)$ on two qubits (see Fig.7 in
     \cite{vatan_optimal_2004}). The circuit $U_c$ is obtained by
     setting $V_1=V_2=W_1=W_2=\mathbb{1}_2$.}
  \label{wholecircuit}
\end{figure}
  We have three angles for $U_c$ and three angles for each
  local unitary, giving a total of 15 parameters.
For all unitaries $U\in \mathcal{SU}(4)$ we have the 
theorem:
\begin{theorem}\label{th:impossible}
Perfect QPS by a unitary transformation is impossible for all initial pure product states of two qubits.
\end{theorem}
For the full proof we refer to appendix \ref{app_theo1}. Here we give a short
version that is valid if one neglects the sets of measure zero
for which the
PSF is undefined. The proof proceeds by contradiction. Suppose there
exists a $U\in\mathcal{SU}(4)$ that perfectly quantum-phase
synchronizes
all initial pure product states of two qubits. Consider the two
initial states with Bloch vectors defined by
$\theta_1=0, \theta_2=0$ for the first state and $\tilde{\theta}_1=\pi,
\tilde{\theta}_2=0$ for the second state. The
Bloch vectors after the entangling gate $U_c$
are $\bn'_1=\cos(\alpha+\beta)\,{\bf e}_z,
\bn'_2=\cos(\alpha+\beta)\,{\bf e}_z$ for the first state and
$\tilde{\bn}'_1=\cos(\alpha-\beta)\,{\bf e}_z,
\tilde{\bn}'_2=-\cos(\alpha-\beta)\,{\bf e}_z$ for the second. For a
well defined
PSF, i.e. non-vanishing Bloch vectors, either ${\bf n}'_1$,
$\tilde{\bn}'_1$ or ${\bf n}'_2$, $\tilde{\bn}'_2$ are directed
oppositely, while the respective other Bloch vectors are aligned. This
still holds after consecutive local rotations, thus preventing perfect
QPS for both states. The full proof
is constructed along similar lines, but takes into account also
the states of measure zero for which the PSF
is undefined.\\

Given this no-go theorem, what is the best
possible QPS averaged over all initial states? We introduce the average PSF as
\begin{align}
\braket{F(\rho_{1,p}\otimes\rho_{2,p},U)}&=\int_\Omega
d\Omega\, F(\rho_{1,p}\otimes\rho_{2,p},U)\,\label{mPSF}
\end{align}
with
$d\Omega=(4\pi)^{-2}d\Omega_1\,d\Omega_2\,$,
$d\Omega_i=\sin\theta_i\,d\theta_id\varphi_i$, and $\Omega$ is the full
spatial angle for both Bloch vectors ($0\le \theta_i\le \pi$, $0\le
\varphi_i< 2\pi$). We consider a unitary transformation to be
optimal if it maximizes the average PSF over all $U\in SU(4)$. In
spite of initial angles 
for which the PSF is 
undefined, the integral is well-defined:
\begin{theorem}\label{th:zeroset}
The set of pure initial product
states for which $F(\rho_{1,p}\otimes\rho_{2,p},U)$ is undefined is of
measure zero for all $U$.
\end{theorem}
The proof is given in appendix \ref{app_theo2}.
It is based on showing that for
all $U$, undefined PSF leads to
relations between
the initial angles that have to be satisfied, reducing thus the
number of free parameters.  This leads to a set of initial states of
measure zero. perfect QPS implies
$\langle F(\rho_{1,p}\otimes\rho_{2,p},U)\rangle=1$, whereas the
converse is not true,
as there may be other states of measure zero where the PSF is defined
but different from one.
The average phase fidelity of two initial states vanishes when we take
them evenly distributed over the two Bloch spheres,
i.e.~$\braket{F(\rho_{1,p}\otimes\rho_{2,p},\mathbb{1}_4)}=0$.\\

We first analyze the performance of the quantum
circuit $U_c$, parametrized by three angles $\alpha,
\beta$ and $\gamma$:
\begin{theorem}\label{th:Uc}
The unitaries $U_c$ given by
eq.(\ref{eq:Uc}) leave the mean phase fidelity of pure initial
product states invariant,
$\braket{F(\rho_{1,p}\otimes\rho_{2,p},U_c)}=0$.
\end{theorem}
The proof of the theorem  is based on symmetry properties of
$F(\rho_{1,p}\otimes\rho_{2,p},U_c)$:
\begin{proof}
\begin{align}
&\braket{F(\rho_{1,p}\otimes\rho_{2,p},U_c)}= \frac{1}{V}\int_\Omega d\Omega\, F\\
&= \frac{1}{(4\pi)^2}\int_{0}^{2\pi}d\varphi_1\int_{0}^{2\pi}d\varphi_2\int_{0}^{\pi}d\theta_1\sin{\theta_1}\notag\\
&\qquad\int_{0}^{\pi}d\theta_2\sin{\theta_2}F(\varphi_1,\varphi_2,\theta_1,\theta_2)\label{p11}\\
&= \frac{1}{(4\pi)^2}\int_{\pi}^{3\pi}d\varphi_1\int_{0}^{2\pi}d\varphi_2\int_{0}^{\pi}d\theta_1\sin{\theta_1}\notag\\
&\qquad\int_{0}^{\pi}d\theta_2\sin{\theta_2}F(\varphi_1,\varphi_2,\theta_1,\theta_2)\label{p12}\\
&= \frac{1}{(4\pi)^2}\int_{0}^{2\pi}d\varphi_1\int_{0}^{2\pi}d\varphi_2\int_{0}^{\pi}d\theta_1\sin{\theta_1}\notag\\
&\qquad\int_{0}^{\pi}d\theta_2\sin{\theta_2}F(\varphi_1+\pi,\varphi_2,\theta_1,\theta_2)\label{p13}\\
&= \frac{1}{(4\pi)^2}\int_{0}^{2\pi}d\varphi_1\int_{0}^{2\pi}d\varphi_2\int_{0}^{\pi}d\theta_1\sin{(\pi-\theta_1)}\notag\\
&\qquad\int_{0}^{\pi}d\theta_2\sin{(\pi-\theta_2)}F(-\varphi_1+\pi,-\varphi_2,\pi-\theta_1,\pi-\theta_2)\label{p14}\\
&= \frac{1}{(4\pi)^2}\int_{0}^{2\pi}d\varphi_1\int_{0}^{2\pi}d\varphi_2\int_{0}^{\pi}d\theta_1\sin{\theta_1}\notag\\
&\qquad\int_{0}^{\pi}d\theta_2\sin{\theta_2}(-F(\varphi_1,\varphi_2,\theta_1,\theta_2))\label{p15}\\
&= -\braket{F(\rho_{1,p}\otimes\rho_{2,p},U_c)}\\
&\Rightarrow\braket{F(\rho_{1,p}\otimes\rho_{2,p},U_c)}=0.
\end{align}
In line (\ref{p12})
 we use the fact that the $\varphi_1$-integration
goes over a complete period, and thus the integration limits may be
shifted (It does not matter which $\varphi$-integral will be
shifted). In line (\ref{p13}) we transfer the integral shift to the
function. These two steps also can be seen as a preceding
$R_z$-Rotation of qubit one by  an angle $\pi$.  The mean PSF is
invariant with respect to preceding local transformations, as they do
not change the set of initial states. In line (\ref{p14}) we apply a
symmetry transformation
\begin{align}
\varphi_1 &\mapsto -\varphi_1\\
\varphi_2 &\mapsto -\varphi_2\\
\theta_1 &\mapsto \pi-\theta_1\\
\theta_2 &\mapsto \pi-\theta_2,
\end{align}
that, as a whole, does not change the mean PSF. Direct
calculation
gives line ($\ref{p15}$).
\end{proof}

Now consider the most general $U\in \mathcal{SU}(4)$ (see
Fig.\ref{wholecircuit}). The first two
local unitaries
$V_1$ and $V_2$ in $U$ can be absorbed without restriction of
generality into the creation of the uniformly distributed initial
states.
As a consequence, the only possibility of increasing the mean PSF is
to take into account the local unitaries $W_1\otimes W_2$.
According to Lemma \ref{lem:1}, we can directly rotate the reduced
Bloch vectors obtained after applying $U_c$.
Expressing $W_1$ and $W_2$ by an Euler decomposition with z- and y-rotations,
 $R_z(\sigma_i)R_y(\nu_i)R_z(\mu_i)$, we can restrict ourselves to one of the final z-rotations
  $R_z(\sigma_i)$ without loss of generality as the PSF only measures the relative phase.
   We define the general transformation as
    $U_g(\alpha,\beta,\gamma, \mu_1, \mu_2,\nu_1,\nu_2,\sigma_1)\equiv\left(R_z(\sigma_1)R_y(\nu_1)R_z(\mu_1)\right)\otimes\left(R_y(\nu_2)R_z(\mu_2)\right)\,U_c$
    and look numerically for angles $\alpha,\beta,\gamma,
    \mu_1,\mu_2,\nu_1,\nu_2,\sigma_1$
that maximize the mean PSF. Without loss of generality, we can
restrict all angles to the interval
$[0,2\pi)$
as
  $\braket{F(\alpha,\beta,\gamma,\mu_1,\mu_2,\nu_1,\nu_2,\sigma_1)}$
  is  $2\pi$-periodic in all angles. Numerical results obtained from
  $10^4$
gradient ascents for randomly generated initial
  angles of $U_g$ suggest the conditions
  \begin{equation}\label{condmax}
\alpha,\beta,\gamma \in
  \{\frac{\pi}{4},3\frac{\pi}{4},5\frac{\pi}{4},7\frac{\pi}{4}\}
  \end{equation}
for
  maxima, while the remaining angles need not have discrete values.
For all these values the maximal
mean PSF is estimated  numerically as
$\braket{F}\simeq 0.349$.
When restricting ourselves to a subset of optimal transformations, it is
possible to obtain a discrete set of angles also for the  remaining angles.
E.g.~by setting $\sigma_1=\mu_2=\nu_2=0$,
leading to a final rotation $R_z(\mu_1)R_y(\nu_1)$,
the numerical maximization of the mean PSF gives
\begin{align}
\mu_1&=\left\{\begin{array}{cl} \pi/2 & \mbox{ if }\alpha+\gamma\in\{n\pi|n\in\mathbb{N}\}\\3\pi/2 & \mbox{else,}\end{array}\right. \notag \\
\nu_1&=\left\{\begin{array}{cl} 0 & \mbox{ if } \alpha+\beta\in\{n\pi|n\in\mathbb{N}\}\\\pi & \mbox{else,}\end{array}\right.\label{maxvalnew}
\end{align}
in addition to (\ref{condmax}).
As a successive rotation of one qubit's Bloch vector by $\pi$ changes
the  sign of the PSF one easily obtains the mean PSF's minima from the
conditions for maxima by appending such a rotation, e.g. by shifting
$\sigma_1$ from $0$ to $\pi$.
The position of the maxima and minima are supported by the analytically verifiable fact
that the gradient of
$\braket{F}$ with respect to all eight angles of $U_g$
vanishes there, see appendix \ref{app_gradients}. \\
\begin{figure}
 \includegraphics[width=0.57\linewidth]{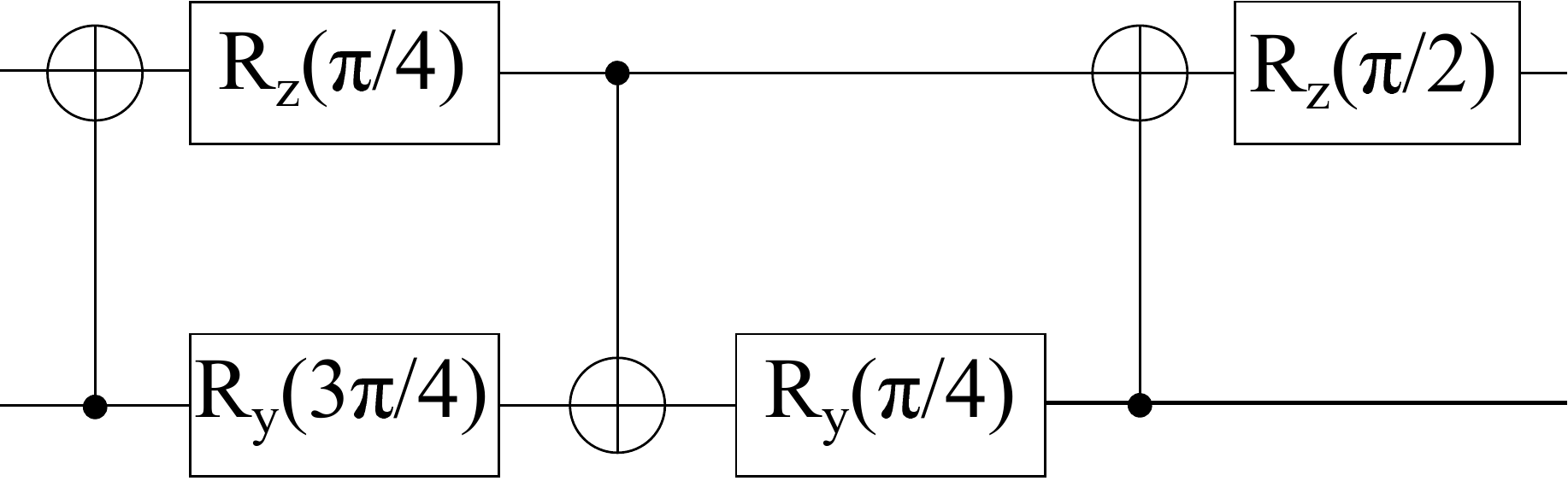}
   \caption{Quantum circuit
     $U_\text{max}$ that maximizes the mean PSF
     for initial pure product states.}
   \label{opt_circuit}
 \end{figure}

We
now examine one of
the unitary transformations that maximizes the PSF,
$U_{\text{max}}=U_{g}(\alpha=3\frac{\pi}{4},\beta=\frac{\pi}{4}
,\gamma=\frac{\pi}{4},\mu_1=\frac{\pi}{2},\nu_1=0)$ in more
detail. The
corresponding optimal quantum circuit is shown in
Fig.\ref{opt_circuit}. The distribution of PSF is shown in
Fig.\ref{dists}, along side with the distribution of $\Delta \varphi$. We
see that $P(\Delta\varphi)$ is symmetric with respect to the
$\Delta\varphi=0$ axis, with two broad maxima in directions close to
$\pm\pi/2$, and a broad minimum for
$\Delta\varphi=\pi$. I.e.~antisynchronization is unlikely, but perfect
QPS is not the most likely outome either.   \\
\begin{figure}
 \includegraphics[width=0.45\linewidth]{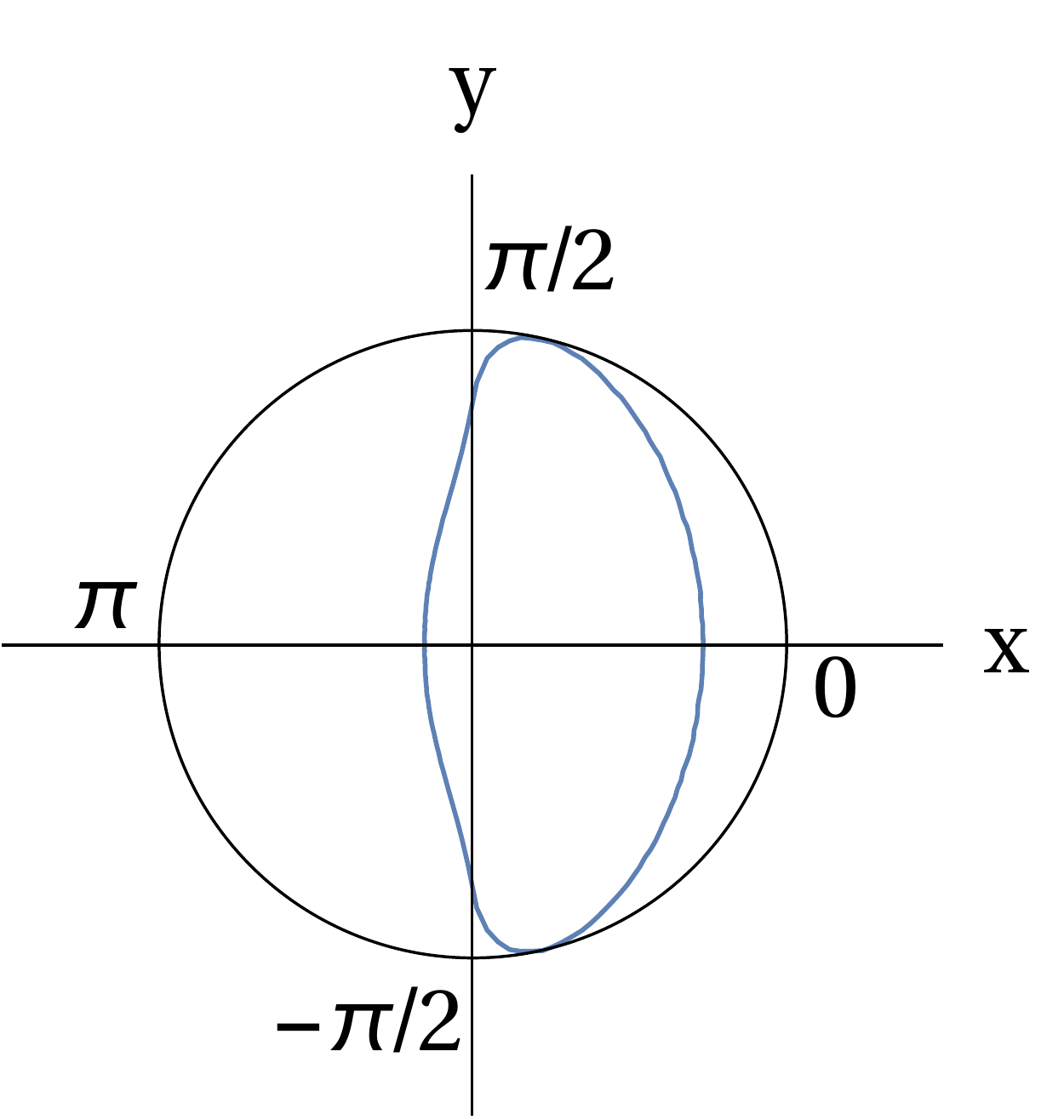}
 \includegraphics[width=0.45\linewidth]{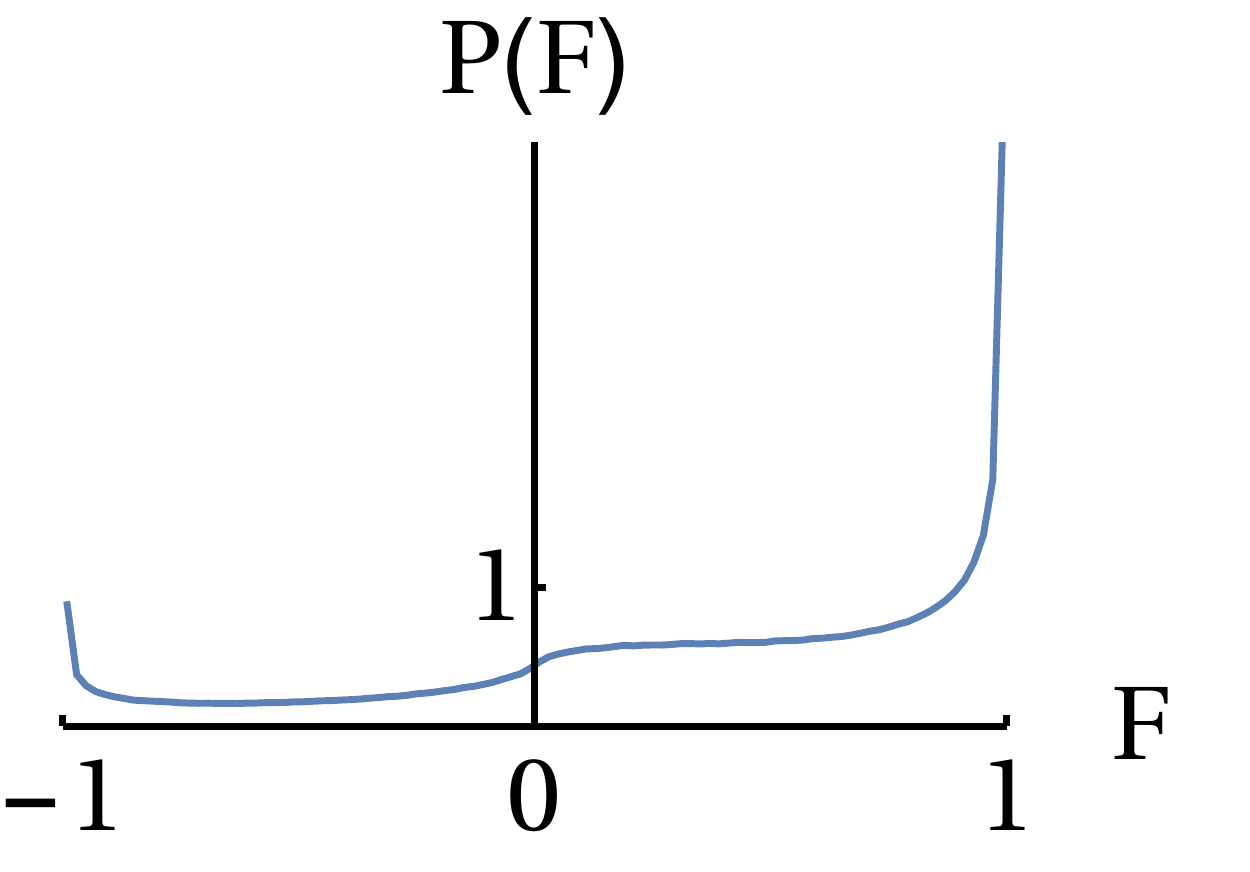}
 \caption{(color online). Distribution of the relative phase
   $P(\Delta\varphi^\prime)$ after the optimal quantum circuit $U_{\text{max}}$ (radial
   coordinate, arbitrary units) as function of $\Delta\varphi^\prime\in(-\pi,\pi]$
   (azimuthal coordinate), i.e. the phase of the second Bloch vector measured relative to the first one, with $\Delta\varphi^\prime=0$ corresponding to the
   $x-$axis (left panel). Corresponding distribution $P(F)$ of the
   PSF, $F=\cos(\Delta\varphi^\prime)$ (right panel). Both distributions are generated numerically.}
 \label{dists}
 \end{figure}
Finally we examine the action of $U_{\rm max}$ on a another subset of
initial states, namely
"equatorial" initial states, with
$\theta_1=\theta_2=\pi/2$ and thus
$n_i=r_i(\cos\varphi_i,\sin\varphi_i,0)$ where $r_i$ is the purity,
  $i\in\{1,2\}$.  These states are important in many applications,
  e.g., linearly polarized photons \cite{bagan_optimal_2006} or the
  BB84 protocol \cite{BB84}.
The transformation $U_\text{max}$ leads to the
transformed Bloch vectors
\begin{eqnarray}
\textbf{n}_1^\prime&=&\frac{1}{2}\left(\begin{array}{c}
     r_1\cos{\varphi_1}-r_2\sin{\varphi_2} \\
     -r_1\sin{\varphi_1}-r_2\cos{\varphi_2} \\
    -r_1r_2\cos(\varphi_1-\varphi_2)
     \end{array}\right),\,\nonumber\\
     \textbf{n}_2^\prime&=&\frac{1}{2}\left(\begin{array}{c}
    r_1\cos{\varphi_1}-r_2\sin{\varphi_2} \\
     -r_1\sin{\varphi_1}-r_2\cos{\varphi_2} \\
    r_1r_2\cos(\varphi_1-\varphi_2)\\
     \end{array}\right).
\end{eqnarray}
The resulting $z$-components are opposite, while the reduced Bloch
vectors are perfectly
synchronized,
$$F\left(\rho_{1}(\theta_1=\frac{\pi}{2})\otimes\rho_{2}(\theta_2=\frac{\pi}{2}),U_\text{max}\right)=1,$$
provided that they are well defined. This means that $U_\text{max}$ achieves
perfect QPS for the subset of equatorial initial states. The same is true for
all transformations satisfying conditions (\ref{condmax}) and
(\ref{maxvalnew}). 
\begin{figure}
 \includegraphics[width=0.9\linewidth]{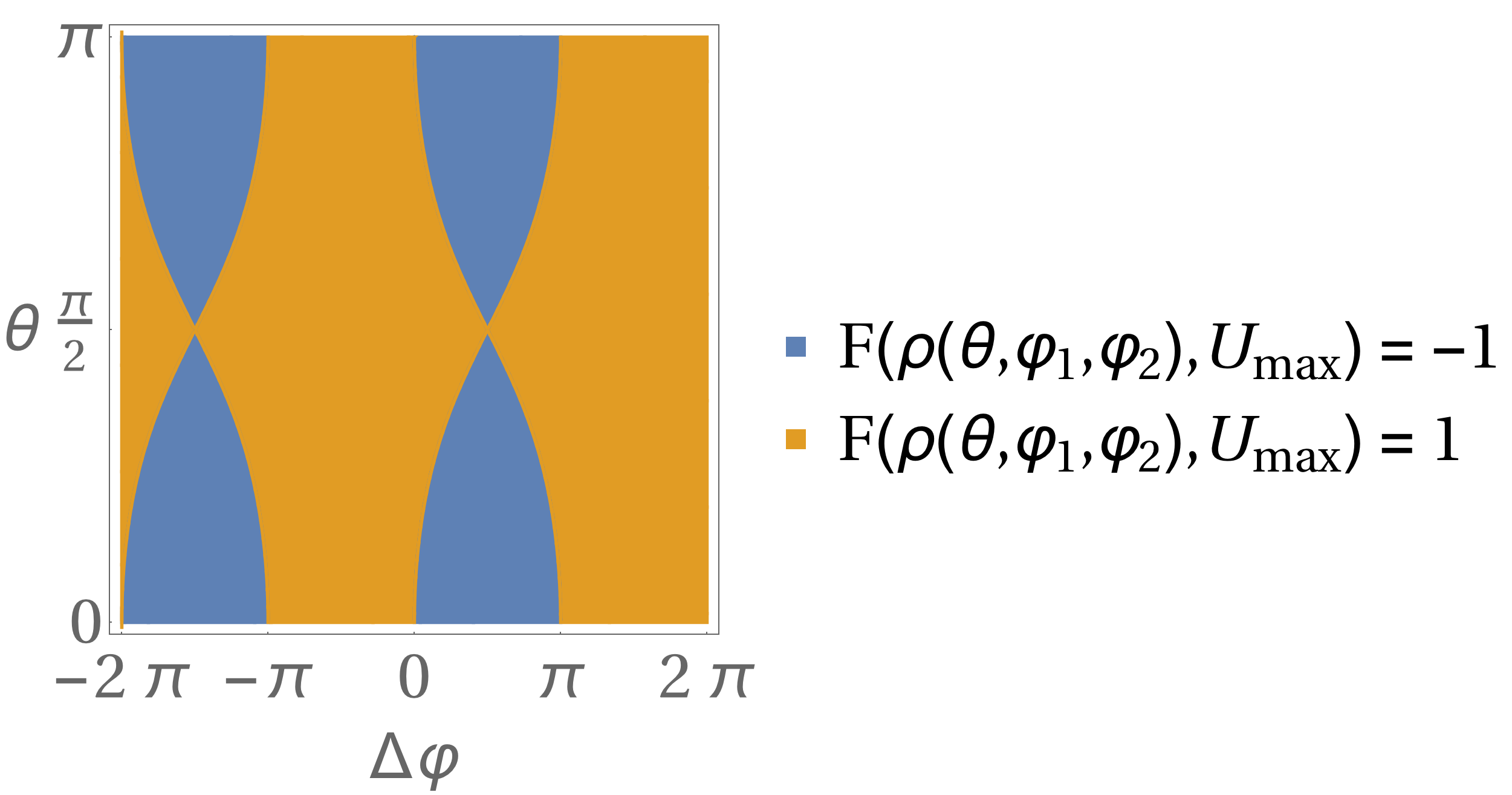}
 \caption{(color online).  PSF for different latitudes $\theta\equiv\theta_1=\theta_2$ as function of $\theta$ and $\Delta\varphi$. While for $\theta=\pi/2$ perfect QPS is achieved, the greater the distance from the equator the worse $U_{\text{max}}$ synchronizes phases.}
 \label{latitudes}
 \end{figure}
One may also wonder about the nature of the final two-qubit state
created by $U_{\rm max}$ and in particular its entanglement. It turns
out that
  the concurrence \cite{nielsen_2000_quantum} of the final state for
  initial pure equatorial
  states ($r_i=1$) is
  $C=\frac{1+\sin(\Delta\varphi)}{2}$. Thus, $U_{\rm max}$
  directly
  encodes the initial relative phase in the final concurrence, such that
  $C=1/2$ corresponds to $\Delta\varphi\in\{0,\pi\}$  and deviations
  from $C=1/2$ are proportional to $\sin(\Delta\varphi)$.
This is by itself an interesting
  property, with possible applications in quantum information
  theory. At the same time it implies that the final entanglement is
  irrelevant for perfect QPS. \\

{\em To summarize,} we have introduced the concept of quantum-phase
synchronization at the example of two qubits.  We have shown that in
contrast to quantum
cloning, 
perfect quantum-phase synchronization of one qubit in an unknown state with $n-1$ qubits in known fixed reference states is possible through joint unitary evolution. For the case of two qubits
  both initially in unknown states, perfect QPS for all initial
states becomes 
impossible through unitary evolution. We have found quantum
circuits that optimize the mean
PSF (averaged over all pure initial product states), and
the distribution
of fidelities and final phase differences for one of the optimal
quantum circuits.
A discrete subset of the optimal
quantum circuits can perfectly quantum-phase synchronize
equatorial initial product
states. Our work opens the road to investigations of quantum-phase
synchronization for larger systems and may find interesting
applications in quantum
information processing. In particular it would be intriguing to see if
phase synchronization can be achieved over distance, and explore
applications in quantum key distribution.\\

\acknowledgements{DB and MK thank the COST 1006 framework for
  support in the initial phase of this project.}
  
\appendix
\widetext
\section{Proof of theorem 1}\label{app_theo1}

The set of initial angles parametrizing initial pure product states is restricted to $\Omega$,
i.e. $0\leq\theta_i\leq\pi,0\leq\varphi_i\le2\pi$ with $i\in\{1,2\}$,
and we simply use the notation $\theta_i,\varphi_i\in\Omega$.
Eight angles $\sigma=(\alpha,\beta,\gamma,\mu_i,\nu_i,\sigma_1)\in\mathbb{R}^8$
parametrize the unitary transformation $U(\sigma)$.
It follows a proof by contradiction.\\

\begin{proof}
Let $U(\sigma_{\text{p}})$, $\sigma_{\text{p}}\in\mathbb{R}^8$,
be the unitary transformation that achieves
$F(\rho(\theta_i,\varphi_i),U(\sigma_{\text{p}}))=1$
$\forall\theta_i,\varphi_i\in{\Omega}$ for which
$\{||\bm_1(\theta_i,\varphi_i,\sigma_{\text{p}})||\neq 0\land||\bm_2(\theta_i,\varphi_i,\sigma_{\text{p}})||\neq0\}$.\\

This means that $\forall\theta_i,\varphi_i\in{\Omega}$ the PSF may
either be undefined,
i.e.~$\{||\bm_1(\theta_i,\varphi_i,\Sigma)||=0\lor||\bm_2(\theta_i,\varphi_i,\Sigma)||=0\}$,
or well defined with $F=1$.
Let us consider discrete subsets
$\Omega_j\subset\Omega$ for $\theta_i,\varphi_i$ and let $\Sigma_j$ be
the set of angles parametrizing the unitary transformation that
achieves perfect QPS for $\theta_i,\varphi_i\in\Omega_j$. Then it follows
by assumption that $\sigma_\text{p}\in\cap_j\Sigma_j$. In the following
we find necessary conditions specifying different $\Sigma_j$, $j=\text{I,II...VI}$,
and obtain a contradiction
by showing that conditions from different $\Sigma_j$ are incompatible,
i.e.~$\cap_j\Sigma_j=\emptyset$.

To find such conditions we do not consider
Bloch vectors after the whole transformation $U(\sigma)$,
but we consider Bloch vectors after the entangling part $U_c$ of $U$,
see equation (11) in the main text. It is worth recalling that according
to lemma 1 final local transformations can be taken into account
by directly rotating the Bloch vectors. As final local transformations are
decomposed in $z$- and $y$-rotations, $\left(R_z(\sigma_1)R_y(\nu_1)R_z(\mu_1)\right)\otimes\left(R_y(\nu_2)R_z(\mu_2)\right)$,
their effect on the Bloch vectors can be taken easily into account.

Note, that oppositely
directed Bloch vectors differ in their phase by $\pi $, or their
phases are undefined. Aligned Bloch vectors have the same phase, given
it is well defined. Remarkably, this does not change after synchronous
rotations. Thereby, two initial states leading to opposite Bloch
vectors for qubit one and aligned Bloch vectors for qubit two (or vice versa) are
particularly useful: $F=1$ is impossible as first Bloch vectors
exhibit identical rotations as well as second Bloch vectors. Thus, at least
one of the initial states has to lead to an undefined PSF.

To simplify notation we define for an arbitrary angle $\delta$ the
corresponding set $S_\delta\equiv\{\delta+n\pi|n\in\mathbb{Z}\}$
that contains all angles modulo $\pi$.
Further, we use $S\equiv S_0\cup S_{\pi/2}$.

Bloch vectors after $U_c(\alpha,\beta,\gamma)$ are given by
\small\begin{equation}
\textbf{n}_1=\left(\begin{array}{c}
\cos \gamma  (\cos \alpha  \sin \theta_2 \cos \varphi_2+\sin \alpha  \sin \theta_1 \cos \theta_2 \cos \varphi_1)-\sin \gamma (\cos \alpha
    \cos \theta_1 \sin \theta_2 \sin \varphi_2+\sin \alpha  \sin \theta_1 \sin \varphi_1)\\
   \sin \gamma  (\cos \beta \cos \theta_1 \sin \theta_2 \cos \varphi_2-\sin \beta \sin \theta_1
   \cos \varphi_1)+\cos \gamma  (\cos \beta \sin \theta_2 \sin \varphi_2-\sin \beta \sin \theta_1 \cos \theta_2
   \sin \varphi_1)\\
   \cos \alpha (\cos \beta \cos \theta_2+\sin \beta \sin \theta_1 \sin \theta_2 \sin \varphi_1 \sin \varphi_2)-\sin \alpha (\cos \beta \sin \theta_1 \sin \theta_2 \cos \varphi_1 \cos
   \varphi_2+\sin \beta \cos \theta_1)\\
\end{array}\right),
\end{equation}

\begin{equation}
\textbf{n}_2=\left(\begin{array}{c}
\cos \gamma (\cos \beta \sin \theta_1 \cos \varphi_1+\sin \beta \cos \theta_1 \sin \theta_2 \cos \varphi_2)-\sin \gamma (\cos \beta \sin \theta_1 \cos \theta_2 \sin \varphi_1+\sin \beta \sin\theta_2 \sin \varphi_2)\\
   \sin \gamma (\cos \alpha \sin \theta_1
   \cos \theta_2 \cos \varphi_1-\sin \alpha \sin \theta_2 \cos \varphi_2)+\cos \gamma (\cos \alpha \sin \theta_1 \sin \varphi_1-\sin \alpha \cos \theta_1 \sin \theta_2
   \sin \varphi_2)\\
   \cos \alpha (\cos \beta \cos \theta_1-\sin \beta \sin\theta_1 \sin \theta_2 \cos \varphi_1 \cos
   \varphi_2)+\sin \alpha (\cos \beta \sin \theta_1 \sin \theta_2 \sin \varphi_1 \sin \varphi_2-\sin \beta
   \cos \theta_2)\\
\end{array}\right).
\end{equation}\normalsize

\textbf{I)}\\

Let $\Omega_\text{I}$ consist of two sets of initial angles $\theta_1=\theta_2=0$ for state 1 and $\theta_1=\pi,\theta_2=0$ for state 2 that (after the transformation $U_c(\alpha,\beta,\gamma)$) lead to Bloch vectors

\begin{minipage}[t]{0.47\textwidth}
	\underline{state 1}: $\theta_1=\theta_2=0$:\\
$n_1=(0,0,\cos (\alpha +\beta ))$\\
$n_2=(0,0,\cos (\alpha +\beta ))$\\
	\end{minipage}
	\hfill
	\begin{minipage}[t]{0.47\textwidth}
		\underline{state 2}: $\theta_1=\pi$, $\theta_2=0$:\\
$n_1=(0,0,\cos(\alpha-\beta ))$\\
$n_2=(0,0, -\cos(\alpha -\beta ))$\\
	\end{minipage}\\
To derive conditions specifying $\Sigma_\text{I}$ it is appropriate to
treat cases of vanishing Bloch vector components  separately. For
$\alpha+\beta\in S_{\pi/2}$ and $\alpha-\beta\notin S_{\pi/2}$, final
Bloch vectors of state 1 are zero, corresponding to maximally mixed sub-states
of qubit one and two. This leads independently from local transformations
to an undefined PSF for state 1, while Bloch vectors of state 2 have non-vanishing,
oppositely directed $z$-components. Undefined PSF is obtained when at least
one Bloch vectors remains on the $z$-axis. This allows arbitrary $z$-rotations
($\mu_i,\sigma_1$) while $\nu_1\in S_0$ or $\nu_2\in S_0$ for $y$-rotations.
$F=1$ requires to rotate Bloch vectors away from the $z$-axis by $\nu_1\notin S_0$
and $\nu_2\notin S_0$, ensuring that PSF is well defined. Then, after $y$-rotations,
Bloch vectors of state 2 lie in the $x$-$z$-plane and, thus, having synchronized
or anti-synchronized phases. Already synchronized phases require $\sigma_1=0$
modulo $2\pi$ to not destroy phase synchronization, while anti-synchronized
phases require $\sigma_1=\pi$ modulo $2\pi$. Altogether, this
restricts $\sigma_1$ to $S_0$.

For $\alpha-\beta\in S_{\pi/2}$ and $\alpha +\beta \notin S_{\pi/2}$ Bloch
vectors of state 2 are zero while Bloch vectors of state 1 have non-vanishing
aligned z-components. Bloch vectors of state 1 lead to an undefined PSF if
$\nu_1\in S_0$ or $\nu_2\in S_0$ while $\sigma_1\in S_0$ is required to
obtain $F=1$ similar to above.

For $\alpha+\beta\in S_{\pi/2}$ and $\alpha-\beta\in S_{\pi/2}$ Bloch vectors
of both states are zero, which does not imply further conditions on
$\nu_i$, $\mu_i$ or $\sigma_1$.

For $\alpha+\beta\notin  S_{\pi/2}$ and $\alpha-\beta\notin  S_{\pi/2}$
both states have Bloch vectors on the $z$-axis. Either first Bloch vectors
are aligned and second opposite or vice versa. Thus, $F=1$ is impossible
for both states. Undefined PSF for at least one state is obtained from $\nu_1\in S_0$ or $\nu_2\in S_0$. This actually makes PSF undefined for both states.\\

In the following we summarize conditions in tables as the following (\ref{tab:I}).
Different rows represent different cases while rows correspond to angles for which
exist conditions.  Note that for $\sigma_p\in\Sigma_\text{I}$ only one of the cases
labeled by i)1., i)2., ...iv)2. in table \ref{tab:I} has
to be fulfilled. The first part of labeling, i.e.~i),ii),...,
corresponds to the first column ($\alpha+\beta$ in table \ref{tab:I}),
counting up if conditions for that angle change,
while the second part of the labeling, 1.,2.,..., simply
counts the cases for each i),ii),... . $\overline{X}$ refers to the
complement of the set $X$. Blank table entries indicate that there
is no condition for the corresponding angle and case.
\begin{table}[h!]
\caption{Summarizing all possible cases with conditions on $\Sigma_\text{I}$. The first column labels the cases while other columns specify angles that are restricted to sets which are written in the rows (different cases).}\label{tab:I}
  \begin{center}
  \bgroup
\def\arraystretch{1.5}
    \begin{tabular}{rrccccc}
  \multicolumn{2}{c}{\textbf{I}} & $\alpha+\beta$ & $\alpha-\beta$ & $\nu_1$ & $\nu_2$ & $\sigma_1$\\
\hline
i) & 1. & $S_{\pi/2}$ & $\overline{S}_{\pi/2}$ & $S_0$& & \\
& 2.& " & " & & $S_0$ & \\
& 3.& " & " & & & $S_0$\\
ii) & 1. & $\overline{S}_{\pi/2}$ & $S_{\pi/2}$ & $S_0$& & \\
& 2.& " & " & & $S_0$ & \\
& 3. & " & " & & & $S_0$\\
iii)&  & $S_{\pi/2}$ & $S_{\pi/2}$ & & & \\
iv)& 1. & $\overline{S}_{\pi/2}$ & $\overline{S}_{\pi/2}$ & $S_0$& & \\
&2.& " & " & & $S_0$ & \\
\hline
    \end{tabular}
    \egroup
  \end{center}
\end{table}

\textbf{II)}\\

\noindent Let $\Omega_\text{II}$ consist of two sets of initial angles $\theta_1=\theta_2=\pi/2$, $\varphi_1=0$, $\varphi_2=\pi/2$ for state 3 and  $\theta_1=\theta_2=\pi/2$, $\varphi_1=\pi$, $\varphi_2=\pi/2$ for state 4 that lead to Bloch vectors

\vspace{0.25cm}
\noindent \begin{minipage}[t]{0.47\textwidth}
	\underline{state 3}: $\theta_1=\theta_2=\pi/2, \varphi_1=0, \varphi_2=\pi/2$:\\
$n_1=(0,\cos(\beta+\gamma),0)$\\
$n_2=(\cos(\beta+\gamma),0,0)$\\
	\end{minipage}
	\hfill
	\begin{minipage}[t]{0.47\textwidth}
		\underline{state 4}:  $\theta_1=\theta_2=\pi/2, \varphi_1=\pi, \varphi_2=\pi/2$:\\
$n_1=(0,\cos(\beta-\gamma),0)$\\
$n_2=(-\cos(\beta-\gamma),0,0)$\\
	\end{minipage}\\
For $\beta+\gamma\in S_{\pi/2}$ and $\beta-\gamma\notin S_{\pi/2}$ as well as for $\beta+\gamma\notin S_{\pi/2}$ and $\beta-\gamma\in S_{\pi/2}$ we refrain from giving further conditions.
For $\beta+\gamma\in S_{\pi/2}$ and $\beta-\gamma\in S_{\pi/2}$ Bloch vectors are zero.
For $\beta+\gamma\notin S_{\pi/2}$ and $\beta-\gamma\notin S_{\pi/2}$ both states have Bloch vectors with a non-vanishing component. PSF is undefined if the first or second Bloch vector is mapped onto the $z$-axis. To map the first Bloch vector onto the $z$-axis (which is identical for both first Bloch vectors), the first $z$-rotation has to rotate to the $x$-axis ($\mu_1\in S_{\pi/2}$) such that the $y$-rotation can map the Bloch vector onto the $z$-axis ($\nu_1\in S_{\pi/2}$). Similarly, one finds for the second Bloch vector $\mu_2\in S_0$, $\nu_2\in S_{\pi/2}$. Similar to I), it is impossible to achieve $F=1$ for both states.\\

\textbf{III)}\\

\noindent \begin{minipage}[t]{0.47\textwidth}
	\underline{state 5}: $\theta_1=\theta_2=\pi/2, \varphi_1=\pi/2, \varphi_2=0$:\\
$n_1=(\cos(\alpha+\gamma),0,0)$\\
$n_2=(0,\cos(\alpha+\gamma),0)$\\
	\end{minipage}
	\hfill
	\begin{minipage}[t]{0.47\textwidth}
		\underline{state 6}:  $\theta_1=\theta_2=\pi/2, \varphi_1=\pi/2, \varphi_2=\pi$:\\
$n_1=(-\cos(\alpha-\gamma),0,0)$\\
$n_2=(0,\cos(\alpha-\gamma),0)$\\
	\end{minipage}\\
Exchanging $\alpha$ with $\beta$ and exchanging the first with the second Bloch vectors maps state 3 onto 5 and state 4 onto 6. Similarly, the conditions can be mapped by exchanging $\alpha$ with $\beta$ and by exchanging the local transformations of qubit one and two ($\mu_1\leftrightarrow\mu_2$, $\nu_1\leftrightarrow\nu_2$). Conditions from II and III are summarized in tables \ref{tab:II}.\\

\begin{table}[h!]
\caption{Summarizing cases with conditions on $\Sigma_\text{II}$ (left) and $\Sigma_\text{III}$ (right). The first column labels the cases while other columns specify angles that are restricted to sets which are written in the rows (different cases).}\label{tab:II}
  \begin{center}
  \bgroup
\def\arraystretch{1.5}
\setlength{\tabcolsep}{0.2cm}
    \begin{tabular}{rrcccccc}
  \multicolumn{2}{c}{\textbf{II}} & $\beta+\gamma$ & $\beta-\gamma$ & $\mu_1$ & $\nu_1$ & $\mu_2$ & $\nu_2$\\
\hline
i) & 1. & $S_{\pi/2}$ & $\overline{S}_{\pi/2}$ & & & & \\
ii) & 1. & $\overline{S}_{\pi/2}$ & $S_{\pi/2}$ & & & & \\
iii)&  & $S_{\pi/2}$ & $S_{\pi/2}$ & & & & \\
iv)& 1. & $\overline{S}_{\pi/2}$ & $\overline{S}_{\pi/2}$ & $S_{\pi/2}$ & $S_{\pi/2}$ & &  \\
&2.& " & " & &  & $S_0$ & $S_{\pi/2}$ \\
\hline
    \end{tabular}
    \qquad\qquad
\begin{tabular}{rrcccccc}
  \multicolumn{2}{c}{\textbf{III}} & $\alpha+\gamma$ & $\alpha-\gamma$ & $\mu_1$ & $\nu_1$ & $\mu_2$ & $\nu_2$\\
\hline
i) & 1. & $S_{\pi/2}$ & $\overline{S}_{\pi/2}$ & & & & \\
ii) & 1. & $\overline{S}_{\pi/2}$ & $S_{\pi/2}$ & & & & \\
iii)&  & $S_{\pi/2}$ & $S_{\pi/2}$ & & & & \\
iv)&1. & $\overline{S}_{\pi/2}$ & $\overline{S}_{\pi/2}$ &  &  &$S_{\pi/2}$ & $S_{\pi/2}$ \\
& 2. & " & "  &  $S_0$ &  $S_{\pi/2}$ & &\\

\hline
    \end{tabular}
    \egroup
  \end{center}
\end{table}

\textbf{IV)}\\

\noindent Let $\Omega_\text{IV}$ contain 16 sets of initial angles parametrizing states 7-22. We look at them as groups of four that are connected by maps allowing to infer conditions for further groups from conditions of the first group. The first group contains states 7-10,

\vspace{0.25cm}
\noindent\begin{minipage}[t]{0.47\textwidth}
	\underline{state 7}: $\theta_1=\pi, \theta_2=\pi/2, \varphi_2=\gamma$:\\
$n_1=(\cos\alpha,0,\sin\alpha\sin\beta)$\\
$n_2=(-\sin\beta,0,-\cos\alpha\cos\beta)$\\
	\end{minipage}
	\hfill
	\begin{minipage}[t]{0.47\textwidth}
		\underline{state 8}:  $\theta_1=0, \theta_2=\pi/2, \varphi_2=-\gamma$:\\
$n_1=(\cos\alpha,0,-\sin\alpha\sin\beta)$\\
$n_2=(\sin\beta,0,\cos\alpha\cos\beta)$\\
	\end{minipage}\\
\begin{minipage}[t]{0.47\textwidth}
	\underline{state 9}: $\theta_1=0, \theta_2=\pi/2, \varphi_2=\pi-\gamma$:\\
$n_1=(-\cos\alpha,0,-\sin\alpha\sin\beta)$\\
$n_2=(-\sin\beta,0,\cos\alpha\cos\beta)$\\
	\end{minipage}
	\hfill
	\begin{minipage}[t]{0.47\textwidth}
		\underline{state 10}:  $\theta_1=\pi, \theta_2=\pi/2, \varphi_2=-\pi+\gamma$:\\
$n_1=(-\cos\alpha,0,\sin\alpha\sin\beta)$\\
$n_2=(\sin\beta,0,-\cos\alpha\cos\beta)$\\
	\end{minipage}\\
Note that second Bloch vectors of state 7 and 8 as well as of state 9 and 10 are always directed oppositely (unless states are maximally mixed).

For $\alpha\in S_0$ the first Bloch vectors are
identical for state 7 and 8 (9 and 10) only having a non-vanishing $x$-component.
 Thus $F=1$ is impossible. Mapping the
first Bloch vectors onto the $z$-axis requires
$\mu_1\in S_0$ in order to keep the $y$-component zero, and $\nu_1\in S_{\pi/2}$
in order to map the $x$-component onto the $z$-axis. On the other hand,
mapping the second Bloch vectors onto the
$z$-axis requires the following, dependent on values of $\beta$:
For $\beta\in S_0$ second Bloch vectors lie on the $z$-axis and
stay there if $\nu_2\in S_0$. For $\beta\in S_{\pi/2}$  second
Bloch vectors lie on the $x$-axis and are mapped onto the $z$-axis
by $\mu_2\in S_0$ (keeping the $y$-component zero) and $\nu_2\in S_{\pi/2}$
(rotating $x$-component onto the $z$-axis). For $\beta\notin S$
second Bloch vectors have non-vanishing $x$- and $z$-components,
and to map them onto the $z$-axis we would need
$\mu_2\in S_0$ (keeping the $y$-component zero) and $\nu_2\in S_{\delta}$ with
$\delta\equiv \arctan\left(\frac{n_{2,x}}{n_{2,z}}\right) =\pm\arctan(\frac{\sin\beta}{\cos\alpha\cos\beta})\notin S$, where the minus sign of $\delta$ belongs to states 7 and 8
and the plus to states 9 and 10. Thus, for $\beta\notin S$ second
Bloch vectors can not all be mapped onto the $z$-axis.

For $\alpha\in S_{\pi/2}$ and $\beta\in S_0$ Bloch vectors are zero. For $\alpha\in S_{\pi/2}$ and $\beta\notin S_0$
first Bloch vectors have a non-vanishing $z$-component
while second Bloch vectors have a non-vanishing
$x$-component. Comparing states 7 and 9 first Bloch vectors are oppositely
directed while second Bloch vectors are aligned, implying that
$F=1$ is impossible. PSF is undefined
if $\nu_1\in S_0$, keeping the first Bloch vectors
on the $z$-axis, or if second Bloch vectors are rotated
to the $z$-axis, $\mu_2\in S_0$, $\nu_2\in S_{\pi/2}$.

For $\alpha\notin S$ and $\beta\in S_0$ first Bloch vectors
have a non-vanishing $x$-component and second Bloch
vectors have a non-vanishing $z$-component. For states 7 and 8
first Bloch vectors are aligned while second Bloch vectors
are directed oppositely. $F=1$ is impossible, and PSF is undefined if
for the first Bloch vectors $\mu_1\in S_0$, $\nu_1\in S_{\pi/2}$ or if
for the second Bloch vector $\nu_2\in S_0$.

For $\alpha\notin S$ and $\beta\in S_{\pi/2}$, looking at states 7 and 9
second Bloch vectors are identical while first Bloch vectors are directed
oppositely. Then, $F=1$ is impossible and an undefined PSF is obtained by
$\mu_1\in S_0$, $\nu_1\in S_{\delta_2}$ with $\delta_2\equiv-\arctan\left(\frac{\cos\alpha}{\sin\alpha\sin\beta}\right)\notin S$ for first Bloch vectors or
by $\mu_2\in S_0$, $\nu_2\in S_{\pi/2}$ for second Bloch vectors.
Equally looking at states 8 and 10, second Bloch vectors are identical
while first Bloch vectors are directed oppositely. Then, $F=1$ is impossible
and an undefined PSF is obtained by $\mu_1\in S_0$, $\nu_1\in S_{-\delta_2}$
for first Bloch vectors or by $\mu_2\in S_0$, $\nu_2\in S_{\pi/2}$ for second Bloch vectors.
Thus, as conditions for first Bloch vectors can not be true at the same time,
conditions for second Bloch vectors have to be true for $\alpha\notin S$ and
$\beta\in S_{\pi/2}$.

For $\alpha\notin S$ and $\beta\notin S$
all Bloch vectors have non-vanishing $x$- and $y$-components.
In the following table we compare states pairwise to find necessary conditions
for $F=1$ or an undefined PSF. We use the angles
$\delta_1\equiv-\arctan\left(\frac{\sin\beta}{\cos\alpha\cos\beta}\right)\notin S$ and
$\delta_2$ defined as above.

\begin{table}[h!]
\caption{Summarizing conditions for states 7-10 for the case that $\alpha,\beta\notin S$. For each pair of states PSF needs to be one or undefined.}\label{tab:notinS}
  \begin{center}
  \bgroup
\def\arraystretch{1.5}
\setlength{\tabcolsep}{0.2cm}
    \begin{tabular}{lrcccc}
 states &  &$\mu_1$ & $\nu_1$ & $\mu_2$ & $\nu_2$  \\
\hline
7 \& 8 & $F=1:$ & $S_0$ & $\overline{S}_0$ & & \\
& undefined: & & & $S_0$ & $S_{\delta_1}$ \\
9 \& 10 & $F=1:$ & $S_0$ & $\overline{S}_0$ & & \\
& undefined: & & & $S_0$ & $S_{-\delta_1}$ \\
7 \& 9 & $F=1:$ & &  & $S_0$ & $\overline{S}_0$  \\
& undefined:  & $S_0$ & $S_{\delta_2}$ & & \\
8 \& 10 & $F=1:$ & &  & $S_0$ & $\overline{S}_0$ \\
& undefined: &  $S_0$ & $S_{-\delta_2}$ & & \\
\hline
    \end{tabular}
    \egroup
  \end{center}
\end{table}

It can be seen from table \ref{tab:notinS} that it is not possible
to achieve an undefined PSF for all states. It follows by this that
the conditions $\mu_1,\mu_2\in S_0$ and $\nu_1,\nu_2\notin S_0$
for $F=1$ need to be true for $\alpha,\beta\notin S$.\\

Group two consists of states 11-14,

\noindent\begin{minipage}[t]{0.47\textwidth}
	\underline{state 11}: $\theta_1=\pi, \theta_2=\pi/2, \varphi_2=\pi/2+\gamma$:\\
$n_1=(0,\cos\beta,\sin\alpha\sin\beta)$\\
$n_2=(0,\sin\alpha,-\cos\alpha\cos\beta)$\\
	\end{minipage}
	\hfill
	\begin{minipage}[t]{0.47\textwidth}
	\underline{state 12}: $\theta_1=0, \theta_2=\pi/2, \varphi_2=\pi/2-\gamma$:\\
$n_1=(0,\cos\beta,-\sin\alpha\sin\beta)$\\
$n_2=(0,-\sin\alpha,\cos\alpha\cos\beta)$\\
	\end{minipage}\\
	\begin{minipage}[t]{0.47\textwidth}
		\underline{state 13}:  $\theta_1=0, \theta_2=\pi/2, \varphi_2=-\pi/2-\gamma$:\\
$n_1=(0,-\cos\beta,-\sin\alpha\sin\beta)$\\
$n_2=(0,\sin\alpha,\cos\alpha\cos\beta)$\\
	\end{minipage}
	\hfill
	\begin{minipage}[t]{0.47\textwidth}
		\underline{state 14}:  $\theta_1=\pi, \theta_2=\pi/2, \varphi_2=-\pi/2+\gamma$:\\
$n_1=(0,-\cos\beta,\sin\alpha\sin\beta)$\\
$n_2=(0,-\sin\alpha,-\cos\alpha\cos\beta)$\\
	\end{minipage}\\
To map states 7 onto 11, 8 onto 12, 9 onto 13 and 10 onto 14 we can exchange $\alpha$ with $\beta$ and rotate after the transformation $U_c$ first Bloch vectors by $R_z(\pi/2)$ and second Bloch vectors by $R_z(-\pi/2)$. Thus, conditions from IV can be mapped onto conditions from V by  exchanging $\alpha\leftrightarrow\beta$ and by shifting $\mu_1\rightarrow\mu_1+\frac{\pi}{2}$ and $\mu_2\rightarrow\mu_2-\frac{\pi}{2}$.\\

Group three consists of states 15-18,

\noindent\begin{minipage}[t]{0.47\textwidth}
	\underline{state 15}: $\theta_1=\pi/2, \theta_2=\pi, \varphi_1=\gamma$:\\
$n_1=(-\sin\alpha,0,-\cos\alpha\cos\beta)$\\
$n_2=(\cos\beta,0,\sin\alpha\sin\beta)$\\
	\end{minipage}
	\hfill
	\begin{minipage}[t]{0.47\textwidth}
	\underline{state 16}: $\theta_1=\pi/2, \theta_2=0, \varphi_1=-\gamma$:\\
$n_1=(\sin\alpha,0,\cos\alpha\cos\beta)$\\
$n_2=(\cos\beta,0,-\sin\alpha\sin\beta)$\\
	\end{minipage}\\
	\begin{minipage}[t]{0.47\textwidth}
		\underline{state 17}:  $\theta_1=\pi/2, \theta_2=0, \varphi_1=\pi-\gamma$:\\
$n_1=(-\sin\alpha,0,\cos\alpha\cos\beta)$\\
$n_2=(-\cos\beta,0,-\sin\alpha\sin\beta)$\\
	\end{minipage}
	\hfill
	\begin{minipage}[t]{0.47\textwidth}
		\underline{state 18}:  $\theta_1=\pi/2, \theta_2=\pi, \varphi_1=\pi+\gamma$:\\
$n_1=(\sin\alpha,0,-\cos\alpha\cos\beta)$\\
$n_2=(-\cos\beta,0,\sin\alpha\sin\beta)$\\
	\end{minipage}\\
As for group two, states of group one are mapped onto states of group three if we exchange $\alpha$ with $\beta$ and exchange the Bloch vectors. Thus, conditions from IV can be mapped onto conditions from VI by the following exchange operations $\alpha\leftrightarrow\beta$, $\mu_1\leftrightarrow\mu_2$, $\nu_1\leftrightarrow\nu_2$ and by mapping $\sigma_1\rightarrow-\sigma_1$.\\
	
Group four consists of states 19-22,

\noindent\begin{minipage}[t]{0.47\textwidth}
	\underline{state 19}: $\theta_1=\pi/2, \theta_2=\pi, \varphi_1=\pi/2+\gamma$:\\
$n_1=(0,\sin\beta,-\cos\alpha\cos\beta)$\\
$n_2=(0,\cos\alpha,\sin\alpha\sin\beta)$\\
	\end{minipage}
	\hfill
		\begin{minipage}[t]{0.47\textwidth}
		\underline{state 20}:  $\theta_1=\pi/2, \theta_2=0, \varphi_1=\pi/2-\gamma$:\\
$n_1=(0,-\sin\beta,\cos\alpha\cos\beta)$\\
$n_2=(0,\cos\alpha,-\sin\alpha\sin\beta)$\\
	\end{minipage}\\
	\begin{minipage}[t]{0.47\textwidth}
	\underline{state 21}: $\theta_1=\pi/2, \theta_2=0, \varphi_1=-\pi/2-\gamma$:\\
$n_1=(0,\sin\beta,\cos\alpha\cos\beta)$\\
$n_2=(0,-\cos\alpha,-\sin\alpha\sin\beta)$\\
	\end{minipage}
	\hfill
	\begin{minipage}[t]{0.47\textwidth}
		\underline{state 22}:  $\theta_1=\pi/2, \theta_2=\pi, \varphi_1=3\pi/2-\gamma$:\\
$n_1=(0,-\sin\beta,-\cos\alpha\cos\beta)$\\
$n_2=(0,-\cos\alpha,\sin\alpha\sin\beta)$\\
	\end{minipage}
To map states from group one onto states of group four we rotate after the transformation $U_c$ first Bloch vectors by $R_z(\pi/2)$ and second Bloch vectors by $R_z(-\pi/2)$, and then  we exchange Bloch vectors. Thus, conditions from IV can be mapped onto conditions from VII by the following exchange operations $\mu_1\leftrightarrow\mu_2+\frac{\pi}{2}$, $\nu_1\leftrightarrow\nu_2$ and by mapping $\sigma_1\rightarrow-\sigma_1$.\\

Conditions from groups 1-4 are given in the tables \ref{tab:IV} summarizing conditions of $\Sigma_\text{IV}$.\\

\begin{table}[h!]
\caption{Summarizing conditions for various cases of groups 1-4. Combining tables for groups 1-4 one obtains the last table summarizing all possible cases with conditions on $\Sigma_\text{IV}$. The first column labels the cases while other columns specify angles that are restricted to sets which are written in the rows (different cases).}\label{tab:IV}
  \begin{center}
  \bgroup
\def\arraystretch{1.5}
\setlength{\tabcolsep}{0.2cm}
    \begin{tabular}{rrcccccc}
  \multicolumn{2}{c}{\textbf{group 1}} & $\alpha$ & $\beta$ & $\mu_1$ & $\nu_1$  & $\mu_2$ & $\nu_2$\\
\hline
i) & 1. & $S_0$ & $S_0$ & $S_0$ & $S_{\pi/2}$  & & \\
   & 2. & " & " & & &   & $S_0$\\
   & 3. & " & $S_{\pi/2}$ & $S_0$ & $S_{\pi/2}$  & & \\
   & 4. & " & " & &  & $S_0$ & $S_{\pi/2}$\\
   & 5. & " & $\overline{S}$ & $S_0$ & $S_{\pi/2}$  & & \\
ii)& 1. & $S_{\pi/2}$ & $S_0$ &   & & & \\
   & 2. & " & $\overline{S}_0$ &  & $S_0$ & & \\
   & 3. & " & " &  & & $S_0$ & $S_{\pi/2}$\\
iii)& 1.& $\overline{S}$ & $S_0$ & $S_0$ & $S_{\pi/2}$  & & \\
    & 2.& " & " & & &   & $S_0$\\
    & 3.& " & $S_{\pi/2}$ &  &   & $S_0$ & $S_{\pi/2}$ \\
    & 4.& " & $\overline{S}$ & $S_0$  & $\overline{S}_0$ & $S_0$ & $\overline{S}_0$ \\
\hline
    \end{tabular}
\qquad
  \begin{tabular}{rrcccccc}
  \multicolumn{2}{c}{\textbf{group 4}} & $\alpha$ & $\beta$ & $\mu_1$ & $\nu_1$  & $\mu_2$ & $\nu_2$\\
\hline
i) & 1. & $S_0$ & $S_0$ & &   & $S_{\pi/2}$ & $S_{\pi/2}$ \\
   & 2. & " & " &  & $S_0$ &  & \\
   & 3. & " & $S_{\pi/2}$ & &   & $S_{\pi/2}$ & $S_{\pi/2}$ \\
   & 4. & " & " & $S_{\pi/2}$ & $S_{\pi/2}$ &  & \\
   & 5. & " & $\overline{S}$  & &   & $S_{\pi/2}$ & $S_{\pi/2}$  \\
ii)& 1. & $S_{\pi/2}$ & $S_0$ &   & & & \\
   & 2. & " & $\overline{S}_0$ &  & & &  $S_0$\\
   & 3. & " & " &  $S_{\pi/2}$ & $S_{\pi/2}$ & &\\
iii)& 1.& $\overline{S}$ & $S_0$  & &  & $S_{\pi/2}$ & $S_{\pi/2}$\\
    & 2.& " & " &  & $S_0$ & &  \\
    & 3.& " & $S_{\pi/2}$ & $S_{\pi/2}$ & $S_{\pi/2}$  &   & \\
    & 4.& " & $\overline{S}$ & $S_{\pi/2}$  & $\overline{S}_0$ & $S_{\pi/2}$ & $\overline{S}_0$ \\
\hline
    \end{tabular}
\\ \vspace{0.5cm}
    \begin{tabular}{rrcccccc}
  \multicolumn{2}{c}{\textbf{group 2}} & $\alpha$ & $\beta$ & $\mu_1$ & $\nu_1$  & $\mu_2$ & $\nu_2$\\
\hline
i) & 1. & $S_0$ & $S_0$ & $S_{\pi/2}$ & $S_{\pi/2}$ &  & \\
   & 2. & " & " & & &  & $S_0$\\
   & 3. & " & $S_{\pi/2}$ &  & &  &  \\
   & 4. & " & $\overline{S}$  & $S_{\pi/2}$ & $S_{\pi/2}$ & & \\
   & 5. & " & " & & &  & $S_0$\\
ii)& 1. & $S_{\pi/2}$ & $S_0$ & $S_{\pi/2}$ & $S_{\pi/2}$ & & \\
   & 2. & " & " & & & $S_{\pi/2}$ & $S_{\pi/2}$\\
   & 3. & " & $S_{\pi/2}$ &  & $S_0$ & & \\
   & 4. & " & $S_{\pi/2}$ &  &  & $S_{\pi/2}$ & $S_{\pi/2}$ \\
   & 5. & " & $\overline{S}$ &  & & $S_{\pi/2}$ & $S_{\pi/2}$ \\
iii)& 1.& $\overline{S}$ & $S_0$ &  $S_{\pi/2}$ & $S_{\pi/2}$  & & \\
    & 2.& " & $S_{\pi/2}$ &  & $S_0$ & & \\
    & 3.& " & $S_{\pi/2}$ &  & & $S_{\pi/2}$ & $S_{\pi/2}$ \\
    & 4.& " & $\overline{S}$ &  $S_{\pi/2}$ & $\overline{S}_0$ & $S_{\pi/2}$ & $\overline{S}_0$ \\
\hline
    \end{tabular}
\qquad
    \begin{tabular}{rrcccccc}
  \multicolumn{2}{c}{\textbf{group 3}} & $\alpha$ & $\beta$ & $\mu_1$ & $\nu_1$  & $\mu_2$ & $\nu_2$\\
\hline
i) & 1. & $S_0$ & $S_0$  & & & $S_0$ & $S_{\pi/2}$ \\
   & 2. & " & " &  & $S_0$ & & \\
   & 3. & " & $S_{\pi/2}$ &  &   &  &  \\
   & 4. & " & $\overline{S}$   & & & $S_0$ & $S_{\pi/2}$ \\
   & 5. & " & " &  & $S_0$ & & \\
ii)& 1. & $S_{\pi/2}$ & $S_0$ & & & $S_0$ & $S_{\pi/2}$ \\
   & 2. & " & " & $S_0$ & $S_{\pi/2}$ & & \\
   & 3. & " & $S_{\pi/2}$ &  &  & & $S_0$ \\
   & 4. & " & $S_{\pi/2}$ &  $S_0$ & $S_{\pi/2}$ &  & \\
   & 5. & " & $\overline{S}$ & $S_0$ & $S_{\pi/2}$  &  &  \\
iii)& 1.& $\overline{S}$ & $S_0$   & & &  $S_0$ & $S_{\pi/2}$\\
    & 2.& " & $S_{\pi/2}$ &  &  & & $S_0$\\
    & 3.& " & $S_{\pi/2}$ & $S_0$ & $S_{\pi/2}$ & &  \\
    & 4.& " & $\overline{S}$ & $S_0$  & $\overline{S}_0$ & $S_0$ & $\overline{S}_0$ \\
\hline
    \end{tabular}
\\ \vspace{0.5cm}
    \begin{tabular}{rrcccccc}
  \multicolumn{2}{c}{\textbf {IV}} & $\alpha$ & $\beta$ & $\mu_1$ & $\nu_1$ & $\mu_2$ & $\nu_2$\\
\hline
i) & 1. & $S_0$ & $S_0$ & & $S_0$ &  & $S_0$\\
   & 2. & " & $S_{\pi/2}$ & $S_0$ & $S_{\pi/2}$ &  $S_{\pi/2}$ & $S_{\pi/2}$ \\
   & 3. & " & " & $S_{\pi/2}$ & $S_{\pi/2}$ &  $S_0$ & $S_{\pi/2}$\\
ii)& 1. & $S_{\pi/2}$ & $S_0$ & $S_{\pi/2}$ & $S_{\pi/2}$  & $S_0$ & $S_{\pi/2}$ \\
   & 2. & " & " & $S_0$ & $S_{\pi/2}$  & $S_{\pi/2}$ & $S_{\pi/2}$\\
   & 3. & " & $S_{\pi/2}$  & & $S_0$ & & $S_0$ \\
\hline
    \end{tabular}

    \egroup
  \end{center}
\end{table}

\textbf{V)}\\

\noindent $\Omega_\text{V}$ refers to states 23 and 24. We do not derive conditions
for $\Sigma_\text{V}$ yet but will refer to these states later, as we then have a certain set of conditions simplifying the handling
of states 23 and 24.

\vspace{0.25cm}
\noindent\begin{minipage}[t]{0.47\textwidth}
\underline{state 23}:  $\theta_1=\pi/2, \theta_2=\pi/2, \varphi_1=\pi/4, \varphi_2=\pi/4$:\vspace{0.2cm}\\
$n_1 =\left(\frac{\cos(\alpha+\gamma)}{\sqrt{2}},\frac{\cos(\beta+\gamma)}{\sqrt{2}},-\frac{\sin(\alpha-\beta)}{2}\right)$\vspace{0.2cm}\\
$n_2 =\left(\frac{\cos(\beta+\gamma)}{\sqrt{2}},\frac{\cos(\alpha+\gamma)}{\sqrt{2}},\frac{\sin(\alpha-\beta)}{2}\right)$\\
	\end{minipage}
	\hfill
	\begin{minipage}[t]{0.47\textwidth}
\underline{state 24}:  $\theta_1=\pi/2, \theta_2=\pi/2, \varphi_1=5\pi/4, \varphi_2=7\pi/4$:\vspace{0.2cm}\\
$n_1 =\left(\frac{\cos(\alpha-\gamma)}{\sqrt{2}},-\frac{\cos(\beta+\gamma)}{\sqrt{2}},\frac{\sin(\alpha+\beta)}{2}\right)$\vspace{0.2cm}\\
$n_2 =\left(-\frac{\cos(\beta+\gamma)}{\sqrt{2}},-\frac{\cos(\alpha-\gamma)}{\sqrt{2}},\frac{\sin(\alpha+\beta)}{2}\right)$\\
	\end{minipage}\\

\textbf{VI)}\\

\noindent $\Omega_\text{VI}$ refers to states 25 and 26. We do not derive conditions
for $\Sigma_\text{VI}$ yet but will refer to these states later, as we then have a certain set of conditions simplifying the handling of states 25 and 26.
	
\vspace{0.25cm}
\noindent\underline{state 25}: $\theta_1=\pi/4,\theta_2=\pi/4, \varphi_1=0,\varphi_2=0$:\vspace{0.2cm}\\
$n_1=\frac{1}{2}\left(\cos\gamma\left(\sqrt{2}\cos\alpha+\sin\alpha\right),\sin\gamma\left(\cos\beta-\sqrt{2}\sin\beta\right),\sqrt{2}\cos(\alpha+\beta)-\cos\beta\sin\alpha\right)$\vspace{0.2cm}\\
$n_2=\frac{1}{2}\left(\cos\gamma\left(\sqrt{2}\cos\beta+\sin\beta\right),\sin\gamma\left(\cos\alpha-\sqrt{2}\sin\alpha\right),\sqrt{2}\cos(\alpha+\beta)-\cos\alpha\sin\beta\right)$\\

\noindent\underline{state 26}: $\theta_1=\pi/4,\theta_2=\pi/4, \varphi_1=\pi,\varphi_2=0$:\vspace{0.2cm}\\
$n_1=\frac{1}{2}\left(\cos\gamma\left(\sqrt{2}\cos\alpha+\sin\alpha\right),-\sin\gamma\left(\cos\beta-\sqrt{2}\sin\beta\right),-\sqrt{2}\cos(\alpha+\beta)+\cos\beta\sin\alpha\right)$\vspace{0.2cm}\\
$n_2=\frac{1}{2}\left(-\cos\gamma\left(\sqrt{2}\cos\beta+\sin\beta\right),\sin\gamma\left(\cos\alpha-\sqrt{2}\sin\alpha\right),-\sqrt{2}\cos(\alpha+\beta)+\cos\alpha\sin\beta\right)$\\

\vspace{1cm}

\noindent The proof proceeds by considering cases of I) one by one. Proving that each case contradicts
other conditions proves a contradiction to the assumption.

First let us consider one of cases I)i) or I)ii) to be true.
From conditions for $\alpha$ and $\beta$ it follows that
\begin{align}
\alpha+\beta\in S_{\pi/2}&,~\alpha-\beta\notin S_{\pi/2}\notag\\
\Rightarrow 2\alpha,2\beta&\notin S_0\notag\\
\Rightarrow\alpha,\beta&\notin S,\label{eq:contr1}\\
\end{align}
Analogously, $\alpha+\beta\notin S_{\pi/2}$, $\alpha-\beta\in S_{\pi/2}$ leads to the same conclusion.
This already contradicts conditions from IV indicating $\alpha,\beta\in S$.\\

Second let us consider I)iii) to be true. This implies
\begin{align}
\alpha\pm\beta\in S_{\pi/2}\Rightarrow 2\alpha,2\beta\in S_0\Rightarrow \alpha,\beta\in S\notag\\
\text{and } \alpha\in S_0\Leftrightarrow\beta\in S_{\pi/2},\notag\\
\text{ as well as } \alpha\in S_{\pi/2}\Leftrightarrow\beta\in S_0.\label{eq:contr5}
\end{align}
It follows that one of cases IV)i)2., IV)i)3., IV)ii)1. or X)ii)2. has to
be true. This corresponds to two possible cases for the local transformations
denoted by A1 and A2,
\begin{align}
\text{A1: }&\mu_1\in S_0,\nu_1\in S_{\pi/2} \text{ and } \mu_2,\nu_2\in S_{\pi/2} \label{eq:contr6}\\
\text{A2: }&\mu_2\in S_0,\nu_2\in S_{\pi/2}\text{ and }\mu_1,\nu_1\in S_{\pi/2}.\label{eq:contr7}
\end{align}

Alternatives i) and ii) from II) and III) lead, similarly to (\ref{eq:contr1}) to $\alpha\notin S$ or $\beta\notin S$, thus, contradicting
conditions (\ref{eq:contr5}). The combination of I)iii), II)iii) and III)iii)
leads to a contradiction for the conditions for $\alpha,\beta,\gamma$,
while the combination of II)iv) and III)iv) contradicts the conditions
(\ref{eq:contr6}) and (\ref{eq:contr7}), respectively.
There remain two combinations, namely II)iii), III)iv), that agree
with A1, and II)iv), III)iii), that agree with A2.
For $\alpha\in S_0,\beta\in S_{\pi/2}$ II)iii) ($\beta\pm\gamma\in S_{\pi/2}$)
implies $\gamma\in S_0$ (A1.1) while III)iii) implies $\gamma\in S_{\pi/2}$ (A2.1),
and for $\alpha\in S_{\pi/2},\beta\in S_0$ II)iii) implies $\gamma\in S_{\pi/2}$
(A1.2) while III)iii) implies $\gamma\in S_0$ (A2.2).

In the following we look at states 23 and 24 from V), doing a
case-by-case analysis for A1.1, A1.2, A2.1 and A2.2:
\begin{itemize}
\item{A1.1:\\
$\alpha,\gamma\in S_0$, $\beta\in S_{\pi/2}$\\
$\mu_1\in S_0$, $\nu_1\in S_{\pi/2}\text{ and }\mu_2,\nu_2\in S_{\pi/2}$\\
These conditions allow us to give the reduced Bloch vectors for state 15 and
state 16 after the transformation
$(R_y(\nu_1)R_z(\mu_1))\otimes(R_y(\nu_2)R_z(\mu_2))\,U_c$.
This means that  up to a remaining z-rotation ($\sigma_1$) all rotations
are already applied to the Bloch vectors.
\begin{align}
\text{state 23:}\notag\\
m_1&=\left(\mp\frac{1}{2}\sin(\alpha-\beta),0\right), m_2=\left(\pm\frac{1}{2}\sin(\alpha-\beta),0\right)\notag\\
\text{state 24:}\notag\\
m'_1&=\left(\pm\frac{1}{2}\sin(\alpha+\beta),0\right), m'_2=\left(\pm\frac{1}{2}\sin(\alpha+\beta),0\right)\notag
\end{align}
The $\pm $ signs are consistent for each component of the first reduced Bloch vectors
as well as for each component of the second reduced Bloch vectors.
Conditions for $\alpha, \beta, \gamma$ give
\begin{equation}
\sin(\alpha-\beta)=-\sin(\alpha+\beta) \Rightarrow m_1=m'_1\text{ and }m_2= -m'_2.
\end{equation}
This means that the second reduced Bloch vectors are opposite while the first
reduced Bloch vectors are identical. It follows that the remaining z-rotation
($\sigma_1$) neither leads to an undefined PSF nor synchronizes both states.}
\item{A1.2:\\
$\alpha\in S_0$, $\beta,\gamma\in S_{\pi/2}$\\
$\mu_1,\nu_1\in S_{\pi/2}\text{ and }\mu_2\in S_0$, $\nu_2\in S_{\pi/2}$ \\
A1.2 leads analogously to A1.1 to the same conclusion as A1.1.}
\item{A2.1:\\
$\alpha,\gamma\in S_{\pi/2}$, $\beta\in S_0$\\
$\mu_1\in S_0$, $\nu_1\in  S_{\pi/2}\text{ and }\mu_2,\nu_2\in S_{\pi/2}$\\
Again, these conditions allow us to give the reduced Bloch vectors for case 15 and case 16 after the transformation $(R_y(\nu_1)R_z(\mu_1))\otimes(R_y(\nu_2)R_z(\mu_2))\,U_c$.
\begin{align}
\text{state 23:}\notag\\
m_1&=\left(\mp\frac{1}{2}\sin(\alpha-\beta),0\right), m_2=\left(\pm\frac{1}{2}\sin(\alpha-\beta),0\right)\notag\\
\text{state 24:}\notag\\
m'_1&=\left(\pm\frac{1}{2}\sin(\alpha+\beta),0\right), m'_2=\left(\pm\frac{1}{2}\sin(\alpha+\beta),0\right)\notag
\end{align}
Conditions for $\alpha,\beta,\gamma$ give
\begin{equation}
\sin(\alpha-\beta)=\sin(\alpha+\beta)\Rightarrow m_1=-m'_1\text{ and }m_2=m'_2.
\end{equation}
As the first reduced Bloch vectors are opposite while the second reduced Bloch
vectors are identical, the remaining z-rotation ($\sigma_1$) neither leads to
an undefined PSF nor synchronizes both states.}
\item{A2.2:\\
$\alpha\in S_{\pi/2}$, $\beta,\gamma\in S_0$ \\
$\mu_1,\nu_1\in S_{\pi/2}\text{ and }\mu_2\in S_0$, $\nu_2\in S_{\pi/2}$\\
A2.2 leads analogously to A2.1 to the same conclusion as A2.1.}
\end{itemize}
Thus, I)iii) leads to a contradiction.\\

Let us consider I)iv) to be true. $\alpha\pm\beta\notin S_{\pi/2}$ already
contradicts most cases of IV). Still possible are the cases IV)i)1.
and IV)ii)3..  Both imply that $\nu_1,\nu_2\in S_0$. I)iv)1. saying
$\nu_1\in S_0$ and I)iv)2. saying $\nu_2\in S_0$ do not add new conditions and can thus be handled together.
IV)i)1. says $\alpha,\beta\in S_0$ which we denote by B1,
while IV)ii)3. says $\alpha,\beta\in S_{\pi/2}$ which we denote by B2.
Similarly to (\ref{eq:contr1}), each of the cases II)i), II)ii), III)i) and III)ii)
implies $\alpha\notin S$ or $\beta\notin S$, thus
contradicting B1 as well as B2.

Let us consider II)iii) and III)iv) to be true: Given B1, it follows from II)iii)
that  $\gamma\in S_{\pi/2}$, which contradicts
$\alpha\pm\gamma\notin S_{\pi/2}$ from III)iv).
Given B2, it follows from II)iii) that $\gamma\in S_0$, again contradicting III)iv).

In analogue fashion one finds a contradiction for II)iv) and III)iii). Thus, there solely
remains II)iii) and III)iii) that imply $\gamma\in S_{\pi/2}$
for B1 and $\gamma\in S_0$ for B2.

In the following we look separately for B1 and B2 at states 25 and 26 from VI).
\begin{itemize}
\item{For B1, we have $\alpha,\beta\in S_0$,
$\gamma\in S_{\pi/2}$ and $\nu_1,\nu_2\in S_0$. Remarkably,
$y$-rotations ($\nu_1,\nu_2$) are restricted such that they
may change the sign of the $x$- and $z$-component or do nothing.
Thus, it suffices to look at the reduced Bloch
vectors after $U_c$. We find
\begin{align}
\text{state 25:}\notag\\
m_1&=\frac{1}{2}\left(0,\sin\gamma\cos\beta\right), m_2=\frac{1}{2}\left(0,\sin\gamma\cos\alpha\right)\notag\\
\text{state 26:}\notag\\
m'_1&=\frac{1}{2}\left(0,-\sin\gamma\cos\beta\right)=-m_1, m'_2=\frac{1}{2}\left(0,\sin\gamma\cos\alpha\right)=m_2.\notag
\end{align}
As first reduced Bloch vectors are oppositely directed and second
reduced Bloch vectors aligned it is always possible to replace the
$y$-rotations (that can change the sign of the $x$-component) by
some rotation around the $z$-axis. Thus, final local rotations
can be seen as one $z$-rotation for each qubit. Due to the orientation
of reduced Bloch vectors $F=1$ is impossible for both states but
PSF is well defined.}
\item{For B2, $\alpha,\beta\in S_{\pi/2}$, $\gamma\in S_0$ and $\nu_1,\nu_2\in S_0$ we find
\begin{align}
\text{state 25:}\notag\\
m_1&=\frac{1}{2}\left(\cos\gamma\sin\alpha,0\right), m_2=\frac{1}{2}\left(\cos\gamma\sin\beta,0\right)\notag\\
\text{state 26:}\notag\\
m'_1&=\frac{1}{2}\left(\cos\gamma\sin\beta,0\right)=m_1, m'_2=\frac{1}{2}\left(-\cos\gamma\sin\beta,0\right)=-m_2,\notag
\end{align}
and, analogously to B1, $F=1$ is impossible for both states but PSF is well defined.
}
\end{itemize}
By this it follows that I)iv) leads to a contradiction, too.
As all cases of I) lead to a contradiction this contradicts the
assumption. Perfect QPS by unitary transformations is impossible.
\end{proof}

\section{Proof of theorem 2}\label{app_theo2}
As well as for the proof of theorem 1, we define the range of initial angles $\theta_i,\varphi_i\in\Omega$ by $0\leq\theta_i\leq\pi,0\leq\varphi_i\le2\pi$ with $i\in\{1,2\}$ and denote the set of angles parametrizing $U$ by $\sigma\equiv(\alpha,\beta,\gamma,\mu_i,\nu_i,\sigma_1)\in\mathbb{R}^8$. We say that a subset $\Omega_0\subset\Omega$ has measure zero if it can be parametrized by less than four parameters.
The PSF is undefined if and only if $\{||\bm_1(\theta_i,\varphi_i,\sigma)||=0\lor||\bm_2(\theta_i,\varphi_i,\sigma)||=0\}$.
\begin{lemma}\label{n1x_lemma}
For any $\sigma$ with $||\bm_1(\theta_i,\varphi_i,\sigma)||=0~\forall\theta_i,\varphi_i\in\Omega_{0,1}\Rightarrow \mu(\Omega_{0,1})=0$.
\end{lemma}
\begin{proof}
\begin{align}
||\bm_1||=0&\Rightarrow m_{1,x}(\theta_i,\varphi_i,\sigma)=0\notag\\
&\Leftrightarrow \sum_{j=1}^{12}c_j(\sigma)t_j(\theta_i,\varphi_i)=0,\label{n1x=0}
\end{align}
with simple trigonometric expressions $t_j(\theta_i,\varphi_i)$ and coefficients $c_j(\sigma)$ given by\\
\begin{minipage}[t]{0.4\textwidth}
\begin{align*}
t_1&=\sin\theta_1\sin\varphi_1\\
t_2&=\cos\theta_1\sin\theta_2\sin\varphi_2\\
t_3&=\sin\theta_1\cos\theta_2\cos\varphi_1\\
t_4&=\sin\theta_2\cos\varphi_2\\
t_5&=\sin\theta_1\cos\varphi_1\\
t_6&=\cos\theta_1\sin\theta_2\cos\varphi_2\\
t_7&=\sin\theta_1\cos\theta_2\sin\varphi_1\\
t_8&=\sin\theta_2\sin\varphi_2\\
t_9&=\cos\theta_1\\
t_{10}&=\sin\theta_1\sin\theta_2\cos\varphi_1\cos\varphi_2\\
t_{11}&=\sin\theta_1\sin\theta_2\sin\varphi_1\sin\varphi_2\\
t_{12}&=\cos\theta_2
\end{align*}
\end{minipage}
\begin{minipage}[t]{0.4\textwidth}
\begin{align*}
c_1&=-a\sin\gamma\sin\alpha\\
c_2&=-a\sin\gamma\cos\alpha\\
c_3&=a\cos\gamma\sin\alpha\\
c_4&=a\cos\gamma\cos\alpha\\
c_5&=b\sin\gamma\sin\beta\\
c_6&=-b\sin\gamma\cos\beta\\
c_7&=b\cos\gamma\sin\beta\\
c_8&=-b\cos\gamma\cos\beta\\
c_9&=-c\sin\alpha\sin\beta\\
c_{10}&=-c\sin\alpha\cos\beta\\
c_{11}&=c\cos\alpha\sin\beta\\
c_{12}&=c\cos\alpha\cos\beta,
\end{align*}
\end{minipage}\\

where the the angles $\mu_1,\nu_1,\sigma_1$, parametrizing the local rotations after $U_c$,  are part of coefficients $a,b,c$ defined by
\begin{align*}
a&\equiv\cos\mu_1\cos\nu_1\cos\sigma_1-\sin\mu_1\sin\sigma_1\\
b&\equiv\sin\mu_1\cos\nu_1\cos\sigma_1+\cos\mu_1\sin\sigma_1\\
c&\equiv\cos\sigma_1\sin\nu_1.
\end{align*}
To prove the lemma we show the opposite direction negated,
\begin{equation*}
\mu(\Omega_{0,1})\neq 0\Rightarrow\nexists~\sigma\text{ with }\sum_{j=1}^{12}c_j(\sigma)t_j(\theta_i,\varphi_i)=0~\forall\theta_i,\varphi_i\in\Omega_{0,1}.
\end{equation*}
The trivial solution is impossible, i.e. $\nexists\sigma$ such that $c_j=0~\forall j$:
\begin{align*}
\text{Suppose that }c_j=0~\forall j\\
&c_j=0, j=1,2,3,4\Rightarrow a=0\\
&c_j=0, j=5,6,7,8\Rightarrow b=0\\
&c_j=0, j=9,10,11,12\Rightarrow c=0.\\
\text{However, }c=0\Rightarrow\\
\text{i) }&\cos\sigma_1=0\Rightarrow a=\mp\sin\mu_1, b=\pm\cos\mu_1\Rightarrow \nexists \mu_1\text{ such that }a=b=0\\
\text{ii) }&\sin\nu_1=0\Rightarrow a=\pm\cos\mu_1\cos\sigma_1-\sin\mu_1\sin\sigma_1 \Rightarrow\\
&1.)~0=\cos\mu_1=\sin\sigma_1\Rightarrow b\neq 0\\
&2.)~0=\sin\mu_1=\cos\sigma_1\Rightarrow b\neq 0\\
&3.)~a=0\Rightarrow 1=\tan\mu_1\tan\sigma_1\text{ and }b=0\Rightarrow -\tan\mu_1=\tan\sigma_1\\
&\quad\Rightarrow 1=-\tan^2\mu_1\text{ which is impossible.}
\end{align*}
As $\mu(\Omega_{0,1})\neq 0$, $\Omega_{0,1}$ provides a range for each parameter that generally consists of arbitrary unions of subintervals (with respect to the intervals of $\Omega$) up to an arbitrary zero set on $\mathbb{R}$. This means that within these subintervals (i.~e. for almost all points of $\Omega_{0,1}$) derivatives by the parameters $\theta_i,\varphi_i$ are well defined.\\
Using the fact that $\sin(x)$ and $\cos(x)$ are reproduced by taking four times the derivation by $x$ we find the following set of equations
\begin{align}
\sum_j c_jt_j&=0 \tag{I}\\
\frac{d^4}{d\theta_1^4}\sum_j c_jt_j&=0 \Rightarrow t_4,t_8,t_{12}\rightarrow 0 \tag{II}\\
\frac{d^4}{d\theta_2^4}\sum_j c_jt_j&=0 \Rightarrow t_1,t_5,t_9\rightarrow 0\tag{III}\\
\frac{d^4}{d\varphi_1^4}\sum_j c_jt_j&=0 \Rightarrow t_2,t_4,t_6,t_8,t_9,t_{12}\rightarrow 0\tag{IV}\\
\frac{d^4}{d\varphi_2^4}\sum_j c_jt_j&=0 \Rightarrow t_1,t_3,t_5,t_7,t_9,t_{12}\rightarrow 0.\tag{V}
\end{align}
Note that a concatenation of the derivations in equations II-V simply combines the corresponding conditions for the $t_j$. In the following we use the notation "(I,II,V)" referring to the concatenation of the derivations in I,II, and V, and we write "I$\pm$II" referring to addition/subtraction of the corresponding equations for the $c_j,t_j$,
\begin{align}
\text{IV$-$(III,IV): } &c_1t_1+c_5t_5=0\label{lem1_eq1}\\
\text{(II,V)$-$(IV,V): } &c_2t_2+c_6t_6=0\label{lem1_eq2}\\
\text{(III,IV)$-$(IV,V): } &c_3t_3+c_7t_7=0\label{lem1_eq3}\\
\text{V$-$(II,V): } &c_4t_4+c_8t_8=0\label{lem1_eq4}\\
\text{(IV,V): } &c_{10}t_{10}+c_{11}t_{11}=0\label{lem1_eq5}\\
\text{I$-$III$-$IV$+$(IV,V): } &c_9t_9=0\label{lem1_eq6}\\
\text{I$-$II$-$V$+$(II,V): } &c_{12}t_{12}=0.\label{lem1_eq7}
\end{align}
We showed that the trivial solution $c_j=0~\forall j$ is impossible. Thus, it can be easily shown that at least one of the equations (\ref{lem1_eq1}-\ref{lem1_eq7}) is a non-trivial equation for the $t_j(\theta_i,\varphi_i)$ that is not true $\forall\theta_i,\varphi_i\in\Omega_{0,1}$.\\
As we showed that for $\mu(\Omega_{0,1})\neq 0$ there do not exist coefficients $c_j(\sigma)$ that satisfy $||\bm_1||=0~\forall\theta_i\varphi_i\in\Omega_{0,1}$, it follows that $||\bm_1||=0 ~\forall\theta_i,\varphi_i\in\Omega_{0,1}\Rightarrow\mu(\Omega_{0,1})=0$, which proves the lemma.
\end{proof}
\begin{lemma}\label{n2y_lemma}
For any $\sigma$ with $||\bm_2(\theta_i,\varphi_i,\sigma)||=0~\forall\theta_i,\varphi_i\in\Omega_{0,2}\Rightarrow \mu(\Omega_{0,2})=0$.
\end{lemma}
\begin{proof}
\begin{align}
||\bm_2||=0&\Rightarrow m_{2,y}(\theta_i,\varphi_i,\sigma)=0\notag\\
&\Leftrightarrow \sum_{j=1}^{8}d_j(\sigma)s_j(\theta_i,\varphi_i)=0.\label{n2y=0}
\end{align}
The proof goes analogue to lemma \ref{n1x_lemma} as the following substitutions are true,
\begin{align*}
d_1&\equiv c_4,d_2\equiv-c_3, d_3\equiv-c_2, d_4\equiv c_1\\
d_5&\equiv-c_8, d_6\equiv c_7,d_7\equiv c_6,d_8\equiv-c_5\\
s_j&=t_j,
\end{align*}
with
\begin{align*}
a&\equiv\cos\mu_2\\
b&\equiv\sin\mu_2\\
c&\equiv 0.
\end{align*}
\end{proof}

Proving \textbf{theorem 2} means to show that for any $\sigma$ with $\{||\bm_1(\theta_i,\varphi_i,\sigma)||=0\lor||\bm_2(\theta_i,\varphi_i,\sigma)||=0\}~\forall\theta_i,\varphi_i\in\Omega_0\Rightarrow\mu(\Omega_0)=0$.
\begin{proof}
According to the lemma \ref{n1x_lemma} $||\bm_1||=0$ is only true for a set of of initial angles $\Omega_{0,1}$ with measure zero, and according to lemma \ref{n2y_lemma} $||\bm_2||=0$ is only true for a set of of initial angles $\Omega_{0,2}$ with measure zero. As $\{||\bm_1||=0\lor||\bm_2||=0\}$ implies that $\Omega_{0}\subset\Omega_{0,1}\cup\Omega_{0,2}$, and as unions of zero sets remain zero sets, it follows that $\Omega_{0}$ has measure zero.\\
\end{proof}
\section{Gradients}\label{app_gradients}
One may verify analytically that at the position of the maxima and
minima the gradient
of the average PSF is zero for any transformation
$U_g$ evaluated at an
extremal set $\Sigma_{\rm ext}$ for the parameters according to
conditions
(14) in the main text (with $\sigma_1=\mu_2=\nu_2=0$),
\begin{align}
\boldsymbol{\nabla}\braket{F(\rho_{1,p}\otimes\rho_{2,p},U_{c,z}(\Sigma))}|_{\Sigma=\Sigma_\text{ext}} &=\braket{\boldsymbol{\nabla}F(\rho_{1,p}\otimes\rho_{2,p},U_{c,z}(\Sigma))(\varphi_1,\varphi_2,\theta_1,\theta_2)}|_{\Sigma=\Sigma_\text{ext}}\\
&=\braket{\boldsymbol{\nabla}F(\rho_{1,p}\otimes\rho_{2,p},U_{c,z}(\Sigma))(\varphi'_1,\varphi'_2,\theta'_1,\theta'_2)\label{p32}}|_{\Sigma=\Sigma_\text{ext}}\\
&=-\braket{\boldsymbol{\nabla}F(\rho_{1,p}\otimes\rho_{2,p},U_{c,z}(\Sigma))(\varphi_1,\varphi_2,\theta_1,\theta_2)}|_{\Sigma=\Sigma_\text{ext}} \label{p33}\\
\Rightarrow\boldsymbol{\nabla}\braket{F(\rho_{1,p}\otimes\rho_{2,p},U_{c,z}(\Sigma))}|_{\Sigma=\Sigma_\text{ext}} &=\boldsymbol{0}\,,
\end{align}
where
$\boldsymbol{\nabla}=(\partial/\partial\alpha,\partial/\partial\beta,\partial/\partial\gamma,\partial/\partial\mu_1,\partial/\partial\nu_1,\partial/\partial\sigma_1,\partial/\partial\mu_2,\partial/\partial\nu_2$).
In line (\ref{p32}) the symmetry transformation
\begin{align}
\theta_1 &\rightarrow\theta'_1=\pi-\theta_1\\
\theta_2 &\rightarrow\theta'_2=\pi-\theta_2
\end{align}
was applied for $\partial/\partial\alpha,\partial/\partial\beta,\partial/\partial\gamma,\partial/\partial\mu_1,\partial/\partial\sigma_1,\partial/\partial\mu_2$-components while the symmetry transformation
\begin{align}
\varphi_1 &\rightarrow\varphi'_1=\pi+\varphi_1\\
\varphi_2 &\rightarrow\varphi'_2=\pi+\varphi_2
\end{align}
was applied for $\partial/\partial\nu_1$- and $\partial/\partial\nu_2$-components.
Calculation gives line (\ref{p33}). To analytically
verify that we have a local maximum (minimum), the Hessian matrix
of second derivatives needs to be shown negative definite (positive
definite), which could not be achieved analytically.\\

\bibliography{bib_phasesynchro}

\end{document}